\documentclass{article}

\usepackage{arxiv}

\usepackage[utf8]{inputenc} % allow utf-8 input
\usepackage[T1]{fontenc}    % use 8-bit T1 fonts
\usepackage{hyperref}       % hyperlinks
\usepackage{url}            % simple URL typesetting
\usepackage{booktabs}       % professional-quality tables
\usepackage{amsfonts}       % blackboard math symbols
\usepackage{nicefrac}       % compact symbols for 1/2, etc.
\usepackage{microtype}      % microtypography
\usepackage{lipsum}		% Can be removed after putting your text content
\usepackage{graphicx}
\usepackage{natbib}
\usepackage{doi}

\usepackage{amssymb}
\usepackage{subfig}
\usepackage{paralist}
\usepackage{url}
\usepackage{comment}
\usepackage{proof}

\usepackage[T1]{fontenc}
\usepackage[scaled=0.80]{beramono}
\usepackage[final]{listings}
\usepackage{multirow}
\usepackage[obeyDraft,textsize=tiny]{todonotes}
\usepackage{listings}
\usepackage{colortbl}
\usepackage{pgfplots}
\usepackage{pgfplotstable}
\usepackage{booktabs}
\usepackage{enumitem}
\usepackage{stfloats}
\usepackage{amssymb}
\usepackage{graphicx}
\usepackage{subfig}
\usepackage{paralist}
\usepackage{url}
\usepackage{proof}
\usepackage{hyperref}

\usepackage[T1]{fontenc}
\usepackage[final]{listings}
\usepackage{colortbl}
\usepackage{booktabs}
\usepackage{enumitem}
\usepackage[euler]{textgreek}

\usepackage{multicol}

\usepackage{siunitx}
\usepackage{multirow}
\usepackage{subfig}
\usepackage{paralist}

\usepackage{amsmath,amsthm}

\usepackage[ruled,vlined]{algorithm2e}
\usepackage{tabularx}

\newcolumntype{L}[1]{>{\raggedright\let\newline\\\arraybackslash\hspace{0pt}}m{#1}}
\newcolumntype{C}[1]{>{\centering\let\newline\\\arraybackslash\hspace{0pt}}m{#1}}
\newcolumntype{R}[1]{>{\raggedleft\let\newline\\\arraybackslash\hspace{0pt}}m{#1}}

\begin{document}
%
%\SetKwRepeat{Do}{do}{while}

%%%%%%%%%%%%%%% article 

% \theoremstyle{definition}
\newtheorem{definition}{Definition}[section]

\newtheorem{theorem}{Theorem}[section]

%%%%%%%%%%%%%%%%article end

\title{Efficient Semantic Summary Graphs for Querying Large Knowledge Graphs}

\author{%
\textsc{Emetis Niazmnad} \\[1ex] % Your name
\normalsize TIB Leibniz Information Centre for Science and Technology, Hannover\\
Leibniz University of Hannover, Germany  \\ % Your institution
\normalsize \href{mailto:Emetis.Niazmand@tib.eu}{Emetis.Niazmand@tib.eu} % Your email address
\and
\textsc{Gezim Sejdiu} \\[1ex] % Second author's name
\normalsize  Smart Data Analytics, University of Bonn, Bonn,
Germany  \\ % Second author's institution
\normalsize \href{mailto:sejdiu@cs.uni-bonn.de}{sejdiu@cs.uni-bonn.de} % Second author's email address
\and % Uncomment if 2 authors are required, duplicate these 4 lines if more
\textsc{Damien Graux} \\[1ex] % Second author's name
\normalsize Inria, Université Côte d’Azur, CNRS, I3S
Sophia Antipolis, France  \\ % Second author's institution
\normalsize \href{mailto:damien.graux@inria.fr}{damien.graux@inria.fr} % Second author's email address
\and 
\textsc{Maria{-}Esther Vidal} \\[1ex] % Second author's name
\normalsize TIB Leibniz Information Centre for Science and Technology, Hannover\\
Leibniz University of Hannover, Germany  \\ % Second author's institution
\normalsize \href{mailto:Maria.Vidal@tib.eu}{Maria.Vidal@tib.eu} % Second author's email address
}

\maketitle              % typeset the header of the contribution
\begin{abstract}
Knowledge Graphs (KGs) integrate heterogeneous data, but one challenge is the development of efficient tools for allowing end users to extract useful insights from these sources of knowledge. In such a context, reducing the size of a Resource Description Framework (RDF) graph while preserving all information can speed up query engines by limiting data shuffle, especially in a distributed setting. This paper presents two algorithms for RDF graph summarization: Grouping Based Summarization (GBS) and Query Based Summarization (QBS). The latter is an optimized and lossless approach for the former method.  We empirically study the effectiveness of the proposed lossless RDF graph summarization to retrieve complete data, by rewriting an RDF Query Language called SPARQL query with fewer triple patterns using a semantic similarity. We conduct our experimental study in instances of four datasets with different sizes. Compared with the state-of-the-art query engine Sparklify executed over the original RDF graphs as a baseline, QBS query execution time is reduced by up to 80\% and the summarized RDF graph is decreased by up to 99\%.
\end{abstract}


\section{Introduction}
\label{intro}
During the past decades, the number of linked datasets -- known as knowledge graphs (KGs)-- has rapidly increased as evidenced in the current state of the Linked Open Data cloud\footnote{As of August 2021, the LOD-cloud gathers around 1,512 datasets sharing 413,734,019,304 RDF Triples. \url{https://lod-cloud.net/}}.
These datasets are structured following the W3C's standard Resource Description Framework, RDF (\cite{rdf}), and share knowledge on various domains, from a more general purpose KGs such as DBpedia (\cite{dbpedia}) or WikiData (\cite{wikidata}) to specialized ones, e.g., SemanGit (\cite{kubitza2019semangit}).
Real-world applications over these types of sources demand the development of optimized techniques to extract meaningful information. 
The Semantic Web community has actively contributed to RDF management and has proposed formalisms, e.g., SPARQL (\cite{harris2013sparql}) and SHACL (\cite{Spahiu2018}), to express queries and integrity constraints over RDF graphs. 
Moreover, within the years, efficiency has also been addressed, and various methods have been proposed; they include methods to store RDF graphs, e.g., centralized (\cite{faye2012survey}) or distributed (\cite{kaoudi2015rdf}), as well as to query RDF graphs (\cite{VidalRLMSP10}). 
Indeed, the task of query processing can become incredibly complex whenever RDF graphs come along with large ontologies, and there may be portions of the ontology with no instances in a knowledge graph. 
Also, complex queries that include graph pattern expressions (e.g., multi-union queries) represent challenges for query engines in processing time (\cite{PerezAG09}). 
Graph summarization is a technique to solve this issue by providing a compact representation of a graph where redundant data is reduced (\cite{ShinG0R19}). As a result, a summarized graph's size is decreased, and effective techniques can be devised to speed up query processing (\cite{kondylakis:hal-02081474}). 
RDF summarization has been used in query answering and optimization. It has been applied to recognizing the most notable nodes, discovering schema from the data, and visualizing the RDF graph to quickly understand the data (\cite{ebiri2018SummarizingSG}). We propose graph summarization methods by applying both word embedding and graph embedding models to find the most similar predicates by encoding them as vectors. The word embedding models resort to Natural Language Processing (NLP) techniques to represent words in a numeric vector space (\cite{DBLP:books/lib/JurafskyM09}). Word embedding models used to convert textual information and social media data such as tweet sentences to numeric weightage in vector format. They are studied in a specific domain to solve real issues, such as \cite{NEOGI2021100019} and \cite{10.3389/fcomp.2021.775368}. There are more use cases which employ word embedding models for vector representation of textual words. \cite{Chauhan2021OptimizationAI} propose a framework which improves the detection of fake news and real news. This framework makes use of neural networks and a tokenization
method. The tokenization
method has been proposed for feature extraction or vectorization, which
assigns tokens to word embeddings. Word embedding models can be applied to RDF graphs as well, and RDF2Vec is an exemplary approach presented by~\cite{DBLP:conf/semweb/RistoskiP16}.

We aim to provide an algorithm in which a summarized RDF graph groups RDF triples composed of similar predicates; the similarity metrics computed over the embeddings determine this relatedness. SPARQL queries are rewritten based on the summarized RDF graph. As a result, query execution time reduced, while answer completeness is maximized. Our goal is to achieve the following research objectives:
\begin{itemize}
    \item Role of summarization in the RDF graph size reduction.
    \item Impact of summarization in query processing. 
\end{itemize}

% \noindent
% \textbf{Approach:} 
Two approaches are presented: Grouping Based Summarization (GBS) and Query Based Summarization (QBS). GBS decreases an RDF graph size and QBS considers criteria of graph summarization to rewrite SPARQL queries into queries with fewer triple patterns, but with equivalent answers. Our query rewriting techniques resort to semantic similarity metrics to identify related predicates in the triple patterns of a SPARQL query and replace them with a predicate that represents all of them. QBS has the following desirable characteristics: \begin{inparaenum}[\bf a\upshape)]\item \textit{Compactness}: graph summarization provides fewer nodes and edges compared with the original RDF graph by considering only a part of the RDF graph which is related to the given SPARQL query; \item \textit{Lossless query processing}: returns the same answers by querying over summarized graph compared with the original one based on similarity metric by transforming a query to the simple one; and \item \textit{Low-cost query processing}: speeds up query processing over the summarized RDF graph. The GBS and QBS performance is evaluated; the Sparklify component (\cite{stadler2019sparklify}) is used as a default query engine from the SANSA Stack (\cite{lehmann-2017-sansa-iswc}). The Waterloo SPARQL Diversity Test Suite (WatDiv) benchmark generator (\cite{alucc2014diversified}) is utilized to generate two RDF graphs (WatDiv.10M and WatDiv.100M) and queries; also, the Entity Summarization Benchmark, ESBM (\cite{ESBM}), and a dump of DBpedia\footnote{\url{https://wiki.dbpedia.org/}} are included in the study. We report on twenty queries where the execution time is accelerated by up to 80\%. The observed results are promising and provide evidence of our proposed approaches' compactness power and their impact on query processing. 
\end{inparaenum}

In particular, the contributions of this work are as follows:
\let\labelitemi\labelitemii
\begin{itemize}
% \addtolength{\itemindent}{1cm}
\item Graph summaries are able to reduce RDF triples required in query processing.
\item Query rewriting techniques guided by RDF graph summarization. These techniques ensure answer completeness.
\item An empirical study over state-of-the-art benchmarks. Observed results indicate the positive effects of reducing redundant information in the portion of an RDF graph required to execute a SPARQL query.
\end{itemize}

The rest of the paper is organized as follows: 
In \autoref{sec:related}, we review the related efforts in the domain of RDF summarization. \autoref{sec:motiv} presents an in-depth example to illustrate our challenges. 
\autoref{sec:approach} presents our proposed approaches. The methodology and the results of our empirical evaluation are reported in \autoref{sec:evaluation}. We have
mentioned our discussions in \autoref{sec:discussion}. Finally, in \autoref{sec:conclusion}, we conclude and draw the next challenges to be addressed.

\section{Related work}
\label{sec:related}
Graph summarization techniques reduce the size of graph, speed up graph query evaluation, as well as facilitate graph visualization and analytics. In addition, it provides semantic searches with a reduction in computational complexity. We analyze existing approaches for RDF graph summarization and query processing over summarized RDF graphs.
%The summarization approaches can be categorized into summarization techniques for graph databases and RDF graphs. 

\subsection{Graph Summarization in RDF}
Graph databases use graph structures for representing entities as nodes and their relationships as edges of a graph (\cite{bourbakis1998artificial}).
The increment of data in graph databases makes the query processing complicated. Summarization technique is a way to overcome the complexity of search query in graph databases (\cite{LeFevreT10}). Graph summarization has been studied for semi-structured graph data models such as XML (\cite{semi-structured}) and RDF graphs. There are many existing works in this area to compact RDF graphs whose structure
is computed from the original RDF graph, such that all the paths present in the original graph are also present in the summary graph (\cite{DBLP:journals/corr/abs-2004-14794}). 
%We should be aware that decreasing the size of the graph does not mean that we will lose the main information from the RDF graph. For this purpose, the method chosen for generating a summary graph plays an important role to return the appropriate results. Various summarizing graph techniques have been proposed with the aim of reducing the size of RDF graphs. 
These techniques can be classified into four categories, as the following (\cite{ebiri2018SummarizingSG}):

\subsubsection{Structural methods}
This summarization method considers the structural RDF graphs. One summarization technique following this method is the adaptive structural summary for RDF graph (ASSG) presented by \cite{ASSG}. It compresses a part of RDF graph which is considered by a collection of queries. 
This technique requires some user-selected queries for building the summary graph. By compressing only the part which consists of the users' queries, the number of edges and nodes is decreased. This technique considers only the structure of the RDF graph and not from a semantic point of view. The Query Based Summarization approach presented in this paper is based on a similar method, but by considering both structure and semantic. Practically, the nodes with the same labels and ranks are assigned in the same equivalence class. Each equivalence class in graph data has a set of nodes, the rank of the nodes, and the labels of the nodes. Therefore, the graph data is divided into some equivalence classes. As a consequence, the compressed graph has fewer nodes and edges compare with the original one. The approach by \cite{10.1145/2630602.2630610} uses structural methods to summarize the RDF graphs and studies efficient query processing in the TriAD (\cite{10.1145/2588555.2610511})
a distributed RDF data management engine, by relying on a summary of the RDF graph stored within the
system. The other approach presented by \cite{Sydow} shows the problem of selecting the most important part of an RDF graph based on a chosen entity by user. This approach lets the system generate a summarization of facts concerning the selective entity where is close to the Query Based Summarization approach presented in this paper.

\subsubsection{Pattern mining methods}
This method discovers patterns to build summary graphs. For example, \cite{Zneika} presents an approach for summarizing RDF graphs using mining a set of approximate
graph patterns and calculating the number of instances covered by each pattern. Then it transforms the patterns to an RDF schema that describes the contents of the knowledge graph. In that case, the evaluation of queries are done over the summarized graph instead of the original graph. Moreover, the computational methods presented by \cite{Karim2020CompactingFS} identify frequent star patterns to generate compact representation of RDF graphs, with a minimized number of frequent star patterns.

\subsubsection{Statistical methods}
This summarization method follows a frequency-based perspective to summarize graphs. The work by \cite{Ghasemi} presents a technique called CoSum where a multi-type graph is as an input and the output is a super-graph. CoSum is assigned to statistical RDF summarization type; it generates summary graphs frequency-based and in quantitatively way. According to the authors, CoSum technique should be used for summarizing the graph by clustering the nodes which share the same type. The idea of grouping subjects with the same predicate and objects in Grouping Based Summarization approach comes from this technique. Then, each cluster refers to the Super-Node which consists of nodes with the same type. These Super-Nodes are linked to each other by weighted edges. Therefore, the main purpose of this approach is to automatically group elements that correspond to the same entity, which is called \textit{Entity Resolution} (\cite{Benjelloun2009}). In general, this technique tries to transform a k-type graph to another k-type summary graph, which consists of Super-Nodes and Super-Edges linked among each other. CoSum as a summarizing technique provides a solution to deal with the following challenges:
\begin{inparaenum}[\bf i\upshape)]
    \item An RDF graph is modeled as a multi-type graph and the collective entity resolution is formulated as a multi-type graph summarization problem.
    \item A multi-type graph co-summarization-based method is proposed in order to identify entities and link connections between them at the same time.
    \item A generic framework is provided to accept different domains-specific knowledge.
\end{inparaenum}
In the summary graph, each Super-Node is a group of some vertices with the same type and each Super-Edge connects these clusters of nodes to each other.

\subsubsection{Hybrid methods}
This method combines two or all other categories to generate summary graphs. As an example, the summarization method presented by \cite{Zheng:2016:SSS:2983200.2983201} is the hybrid RDF summarization method, because it considers both structure and patterns to construct a summary graph.

\subsection{Semantic search and query processing}
\label{sec:SemanticSearch}
Motivated by the aim of simplifying queries consisting of the union of some queries, \cite{Zheng:2016:SSS:2983200.2983201} developed a solution based on similarity search. This means instead of using multiple union of queries to get the complete answers, only by using a single or less union of queries can get the correct and the same results which are given by multi-union queries.
Our work shares the same observation: reducing SPARQL complexity can be achieved thanks to summarization techniques. This technique includes both structural and semantic similarities. Since queries can be different in terms of structure, but they have similar semantic meanings, some operations such as semantic path substitution are introduced. Semantic path substitution operation is used to replace a path
with an edge by mining the structure patterns. A dictionary of semantic instances is provided to mine semantic graph patterns by keeping instances which are semantically equivalent. Finally, rewriting a given query graph by semantic path substitution gives a set of semantically equivalent queries.
Also, based on these operations, a similarity measure, called \textit{Semantic Graph Edit Distance} \textit{(sged)}, is defined by \cite{Zheng:2016:SSS:2983200.2983201}. \textit{Sged} measures the cost of transforming one Sub-Graph to another one. 
Then Sub-Graphs extracted from the RDF graph will be chosen to provide the summary graph if they have minimum \textit{sged-based} transformation cost.

\begin{figure*}[t!]
\centering  
\hspace{0pt}\subfloat[Original RDF Graph]{
     \includegraphics[width=0.5\textwidth]{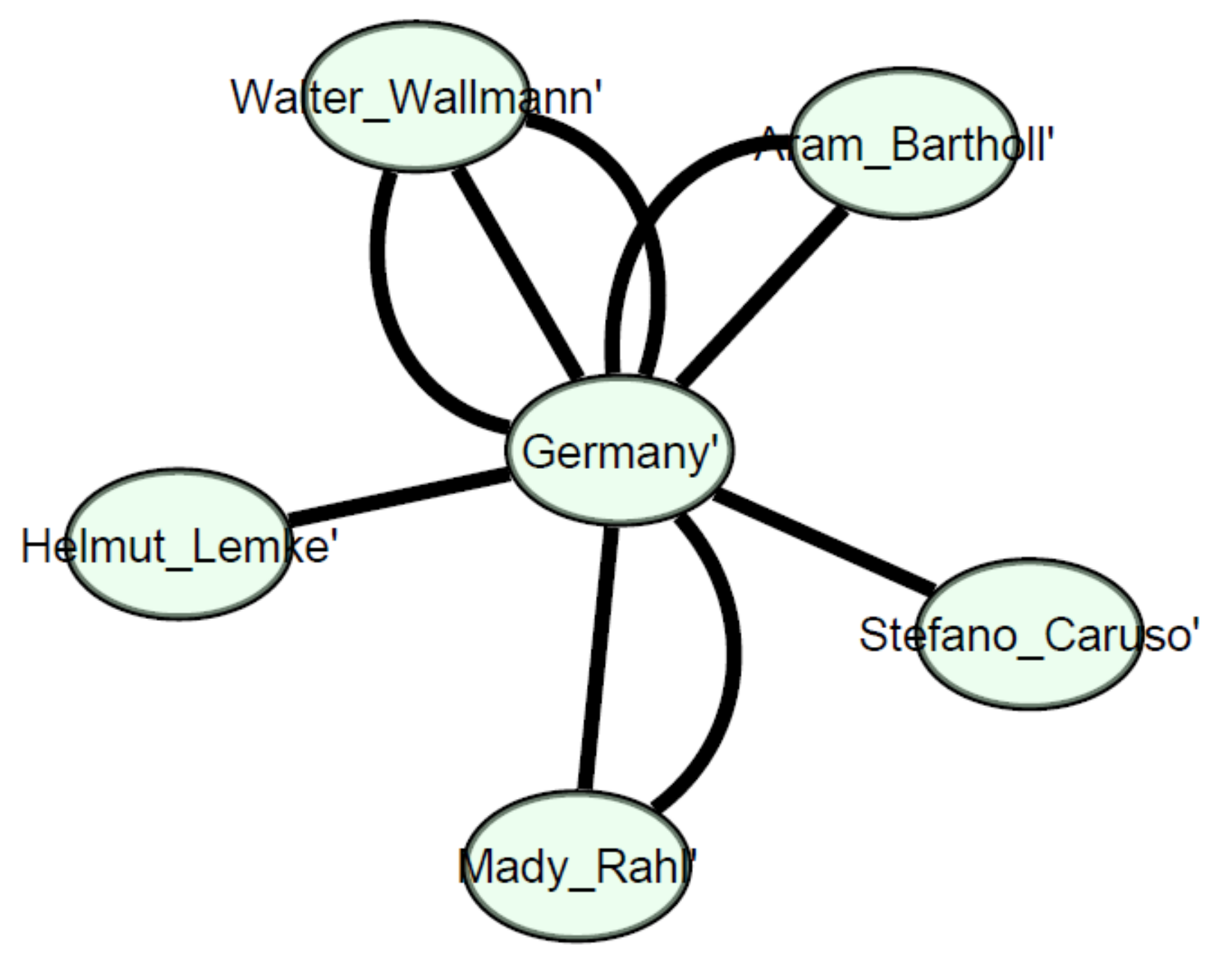}
     \label{fig:1a}}
     \vspace{0pt}\subfloat[Summarized RDF Graph]{
     \includegraphics[width=0.4\textwidth]{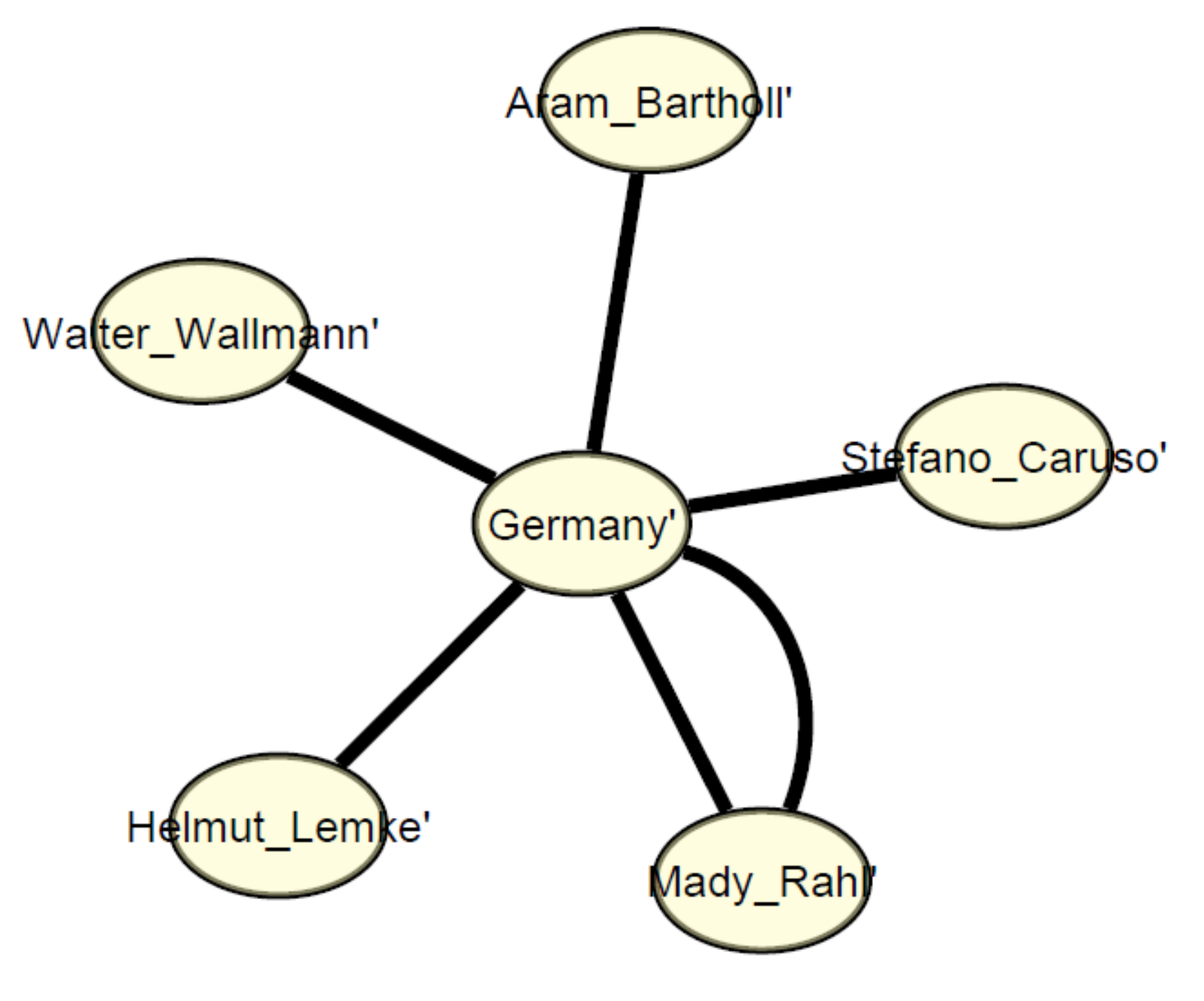}
     \label{fig:1b}}
      \hspace{0pt}\subfloat[Triples summarization]{
     \includegraphics[width=.95\textwidth]{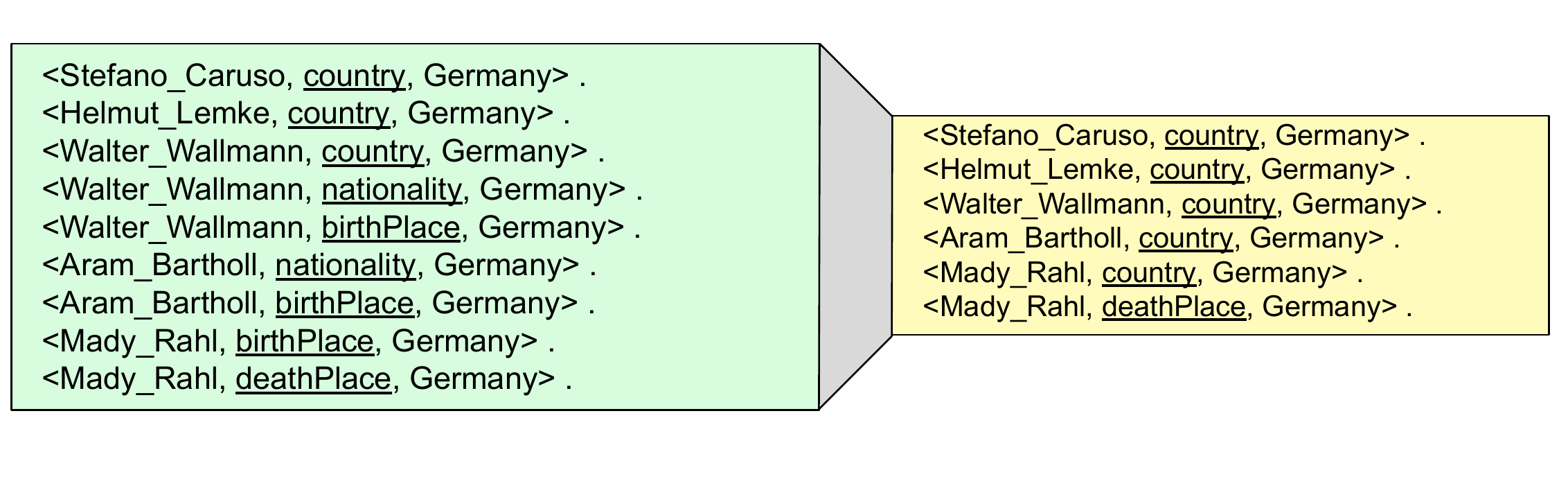}
     \label{fig:1c}}
\caption{{\bfseries Motivating Example (Compactness).} (a) An RDF graph representing entities with similar properties; (b) A lossless summarization of the RDF graph preserving main information; (c) Triples related to the RDF graph in the green box are summarized to a smaller portion in the yellow box. }
 \label{fig:MotiExmp}
\end{figure*}

\begin{figure*}[t!]
\centering  
\hspace{0pt}\subfloat[Original RDF Graph and SPARQL query]{
     \includegraphics[width=1\textwidth]{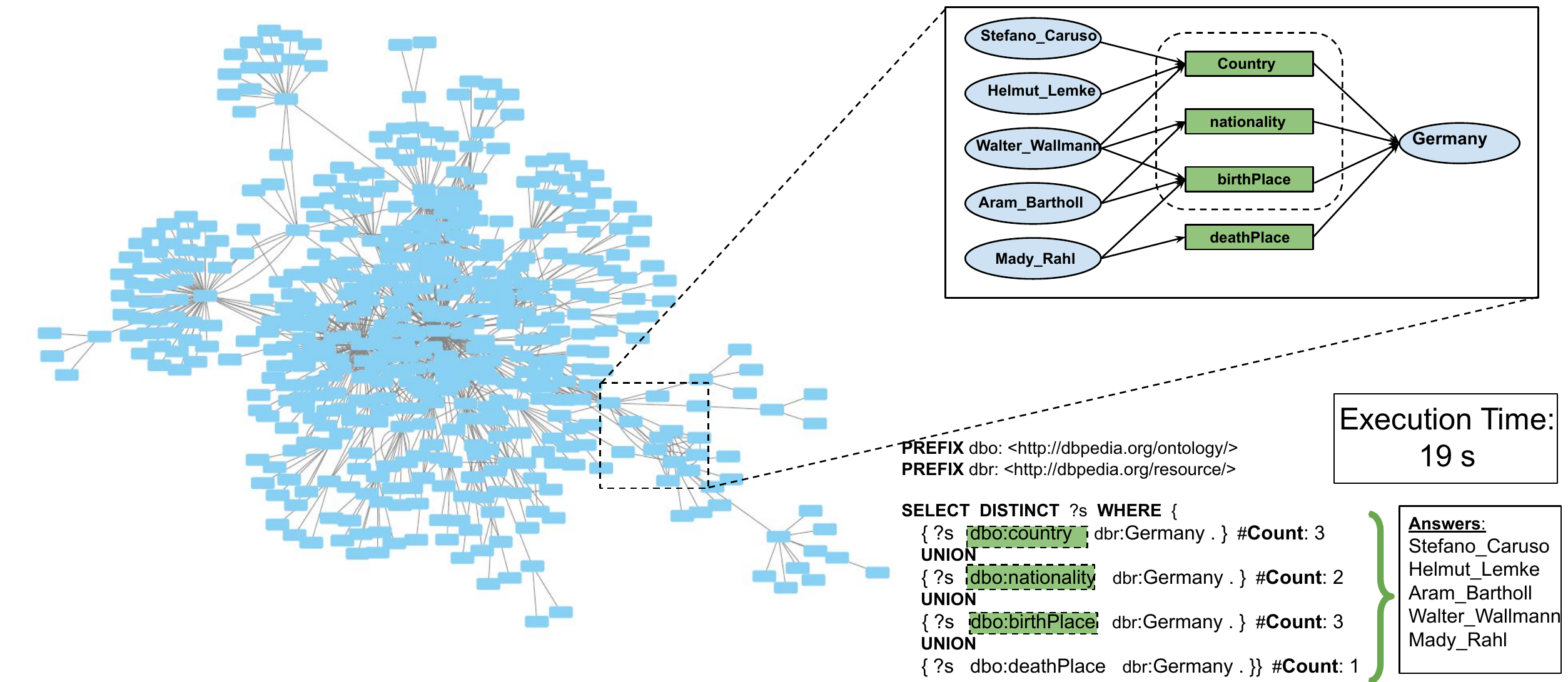}
     \label{fig:dbpediaOrg}}
     \hspace{0pt}\subfloat[A Naive Approach for \\Summarizing RDF Graphs]{
     \includegraphics[width=.5\textwidth]{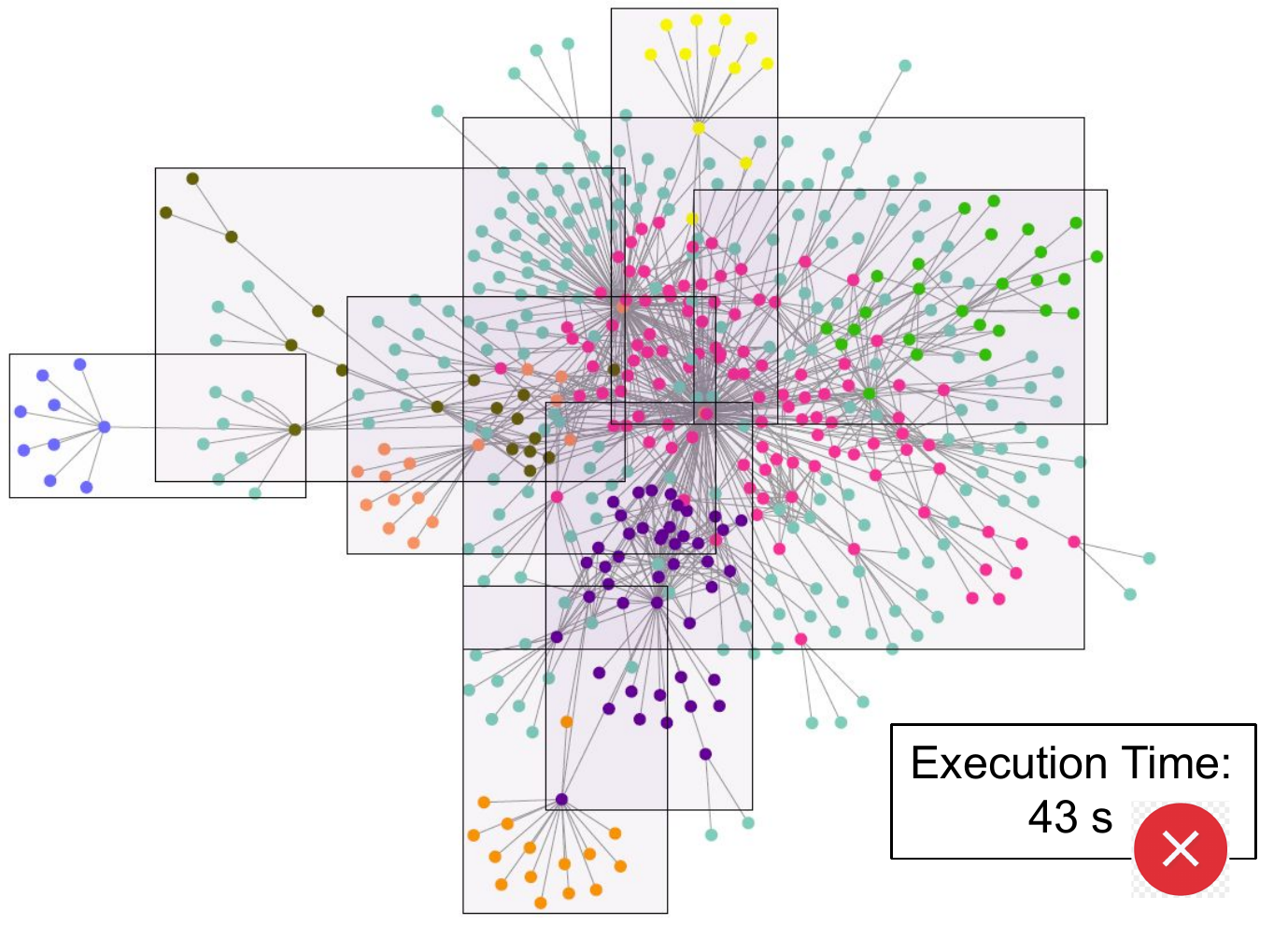}
     \label{fig:dbpediaSumG}}
     \hspace{0pt}\subfloat[An Optimized Approach for \\Summarizing RDF Graphs]{
     \includegraphics[width=.45\textwidth]{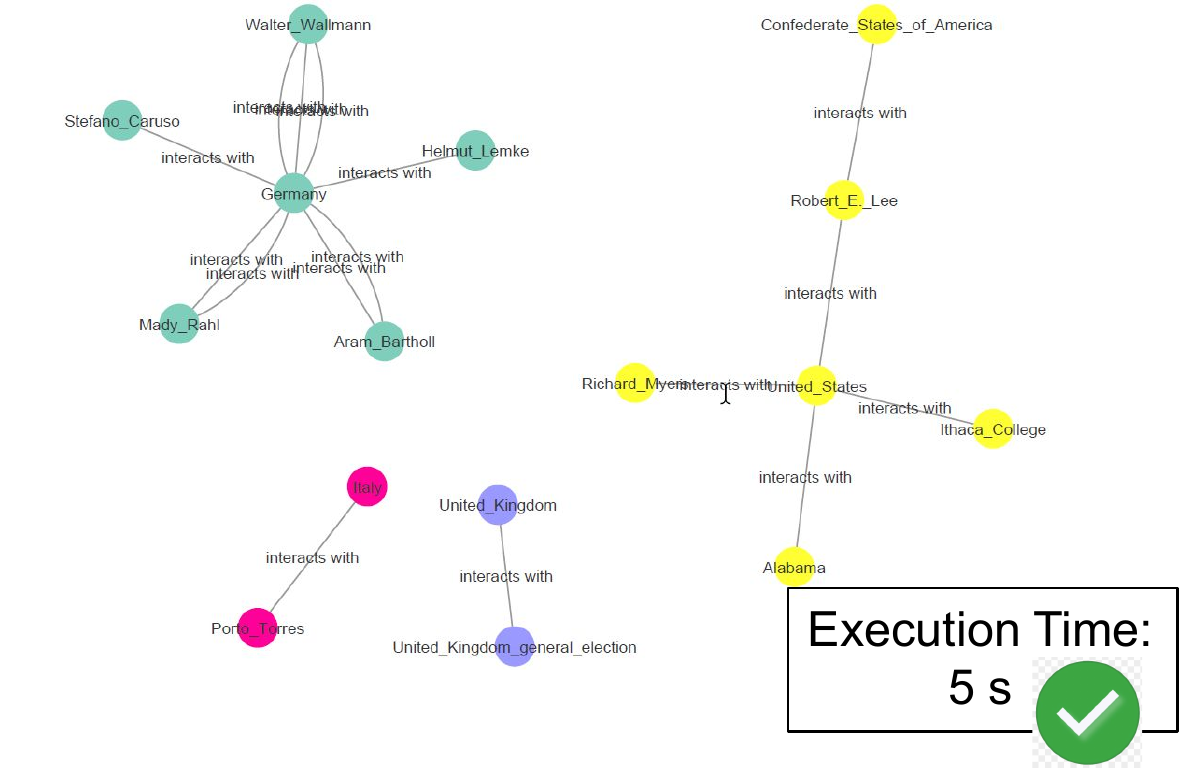}
     \label{fig:dbpediaSum}}
\caption{{\bfseries Motivating Example (Execution Time).} (a) A SPARQL query retrieves five answers in 19 seconds from the original RDF graph; (b) Summarized RDF graph by naive approach retrieves the same answers in 43 seconds; (c) Optimized approach retrieves the same answers in 5 seconds.}
 \label{fig:cytoscape}
\end{figure*}

\section{Challenges and motivation}

\label{sec:motiv}
Efficient query processing over large RDF graphs is one of the main challenges in data management. 
We motivate this data management problem with two examples and illustrate -- with a real-world use case -- the impact of an RDF graph size on execution time. Figure \autoref{fig:1a} depicts a portion of an RDF graph with entities related by properties. It comprises properties that are semantically similar (e.g., $country$, $nationality$, and $birthPlace$). 
Let us consider a graph summarization method by \cite{Ghasemi} that groups similar entities and properties in a graph.
All the elements of a group (i.e., entities or properties) are summarized into one element (i.e., into an entity or a property) in the summarized graph. Some works have been done for grouping the elements with
similar semantic meaning, e.g., \cite{Singh} propose a method to group the terms with similar semantic meaning by evaluating the similarity between words using GloVe (\cite{DBLP:conf/emnlp/PenningtonSM14}).

The results of applying this method to the RDF graph in Figure\autoref{fig:1a} are illustrated in  Figure\autoref{fig:1b}. Moreover, Figure\autoref{fig:1c} presents the RDF serialization of the RDF triples in Figure\autoref{fig:1a} and Figure\autoref{fig:1b}.

A portion of DBpedia that consists of 2,047 RDF triples
is presented in Figure \autoref{fig:dbpediaOrg}. Nodes and edges are represented as ovals and rectangles, respectively. The dataset has 1,000 edges and 510 nodes. This figure also presents a SPARQL query comprising four triple patterns. 
The evaluation of this query retrieves five answers that correspond to the names of people in Germany, or with German nationality, or born or die in Germany. 
The Sparklify query engine~\footnote{\url{http://sansa-stack.net/sparklify/}} produces these answers in 19 seconds. 
Figure \autoref{fig:dbpediaSumG} and Figure \autoref{fig:dbpediaSum} depict the execution time over summarized RDF graphs computed following the two graph summarization methods previously described.  
The results in Figure\autoref{fig:dbpediaSumG} suggest that the execution of the SPARQL query over the summarized graph by naive approach, even producing all the results, can be costly. This approach grouped source nodes with similar edges and target nodes in 174 Sub-Graphs. Thus, query processing over the summarized RDF graph requires 43 secs. to produce the five answers. Alternatively, graph summarization can be done during query processing. An optimized query-based graph summarization method  can identify the portion of the RDF graph required to answer the query, and then it summarizes only this portion of the original RDF graph. The results in Figure\autoref{fig:dbpediaSum} show the query execution over a summarized RDF graph by optimized approach with 16 edges and 17 nodes; it retrieves all the five answers in 5 seconds. The RDF graphs in ~\autoref{fig:cytoscape} are generated by Cytoscape~\footnote{\url{http://www.cytoscape.org/}}.

Figure \autoref{fig:3a} presents a compact representation of the portion of DBpedia in Figure \autoref{fig:dbpediaOrg} required to answer the SPARQL query in Figure \autoref{fig:3b}. This graph represents--with one property-- the three related properties \textit{country}, \textit{nationality}, and \textit{birthPlace}. The property \textit{deathPlace} remains in the compact query since it is not semantically similar to other properties. 
This compact modeling reduces the RDF graph size and enables rewriting the query into a query with fewer triples in Figure \autoref{fig:3b}. As a result, the rewritten query execution time is reduced to 5 seconds, while the complete five answers are produced.
These examples illustrate the relevance of efficient summarization techniques in RDF data management. In \autoref{sec:approach}, we address this problem, and describe summarization techniques able to
reduce the size of RDF graphs and speed up execution time during query processing over summarized graphs.

\begin{figure*}[t!]
\centering  
\vspace{0pt}\subfloat[Summary RDF graph with fewer properties and entities]{
     \includegraphics[width=.45\textwidth]{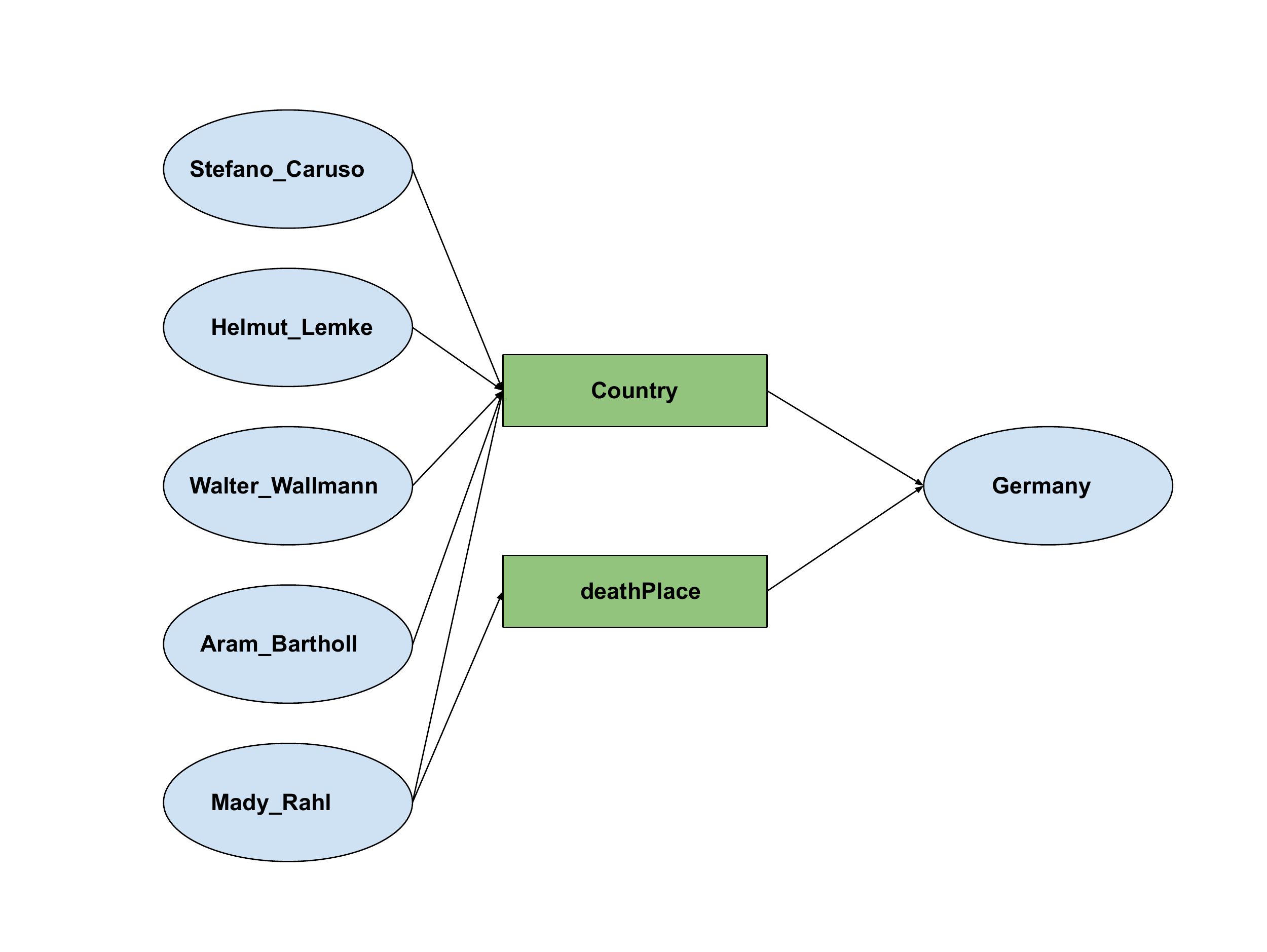}
     \label{fig:3a}}
     \vspace{0pt}\subfloat[Summarized triple pattern query]{
     \includegraphics[width=.55\textwidth]{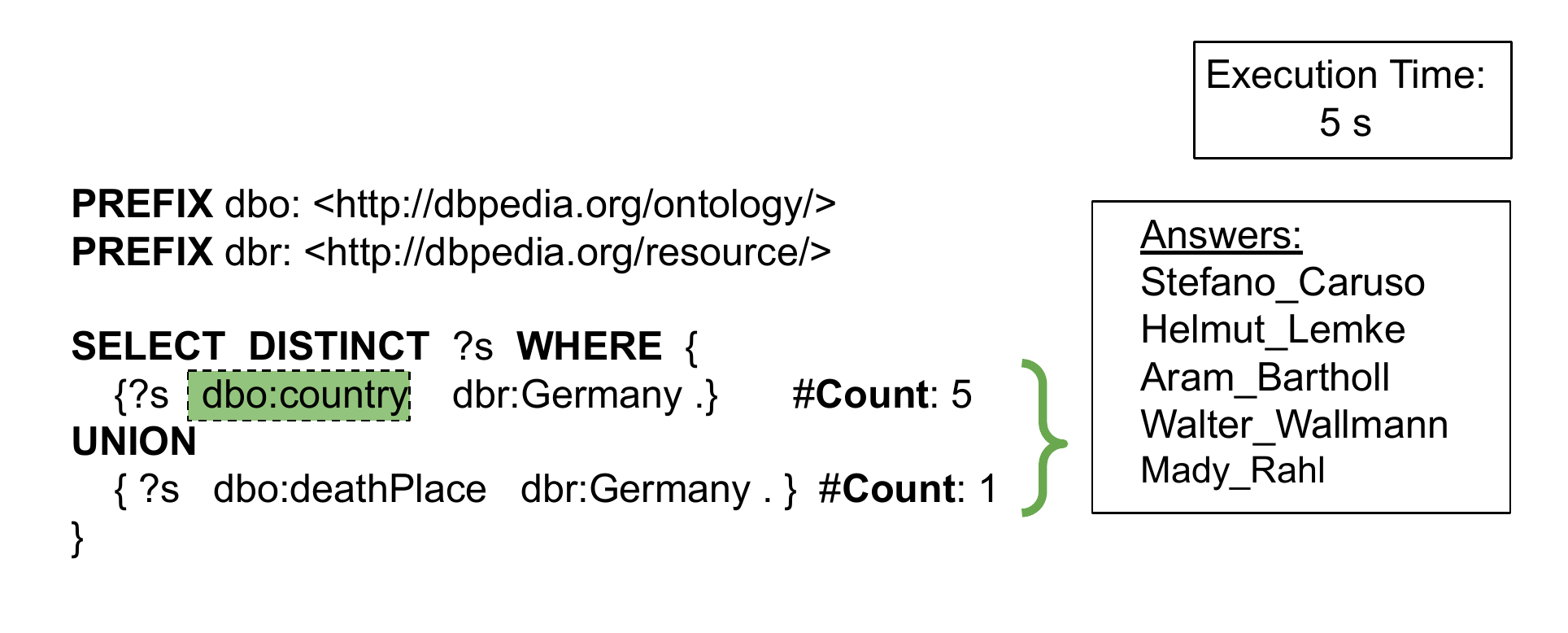}
     \label{fig:3b}}
     
\caption{{\bfseries Example of Summarized RDF.} (a) An RDF graph by considering only one of the most similar predicates returns five answers in 5 seconds. Simplify returning the same answers without being aware of schema and all predicates in less time by reducing the number of triples in RDF dataset; (b) Retrieve the same answers in less time compared with original RDF graph by applying the transformed SPARQL query as a simple one over the summarized RDF graph.}
 \label{Simplequery}
\end{figure*}

\section{Research problem and proposed approach}
\label{sec:approach}
In this section, we discuss the problem of query processing over summarized RDF graphs. Also, we introduce important preliminary definitions and provide the solution by proposing two approaches, naive and optimized, for summarizing RDF graphs without losing necessary information.

\subsection{Problem statement}
Summarization techniques minimize the size of RDF graphs, which helps optimize query processing. Meanwhile, preserving all needed information should be considered during summarization process. Summarizing graph based on a semantic similarity measure is a technique to solve query processing over large RDF graphs. This section presents techniques that exploit knowledge encoded in RDF graphs, similarity measures, and SPARQL queries; they generate summarized graphs against which queries can be processed. To illustrate the relevance of determining relatedness using similarity measures, consider the RDF graph in Figure \ref{fig:1b}; it summarizes the edges of an RDF graph in Figure \ref{fig:1a} that are related or similar. A similarity measure is a function that given two entities associates a value in the range [0.1] that indicates a degree of relatedness of the input entities. Similarity measures rely on various properties of the input entities to estimate similarity. 

\subsection{Preliminaries}
Given two classes $C\textsubscript{1}$=\{$e\textsubscript{11}$,$e\textsubscript{12}$, ...,$e\textsubscript{1n}$\} and $C\textsubscript{2}$=\{$e\textsubscript{21}$,$e\textsubscript{22}$, ...,$e\textsubscript{2m}$\} where $e\textsubscript{1i}$ and $e\textsubscript{2j}$ are entities in classes $C\textsubscript{1}$ and $C\textsubscript{2}$, respectively. The semantic similarity between these two classes is defined as $sim(C\textsubscript{1},C\textsubscript{2})$=$Avg(sem\_sim(e\textsubscript{1i},e\textsubscript{2j}))$ for the set of all pairs of $(e\textsubscript{1i},e\textsubscript{2j})$ in $C\textsubscript{1}$ X $C\textsubscript{2}$ and $e\textsubscript{1i}$ !=$e\textsubscript{2j}$. Therefore, two classes are similar to each other if and only if a set of entities in class $C\textsubscript{1}$ are similar to a set of entities in class $C\textsubscript{2}$ (\cite{semanticsimilarity}). Also, these entities should not be the same. The value of similarity is equal to the average of the semantic similarity value of all entity pairs.  

One of the metrics to measure the similarity is cosine similarity; it calculates the similarity between two n-dimensional vectors by looking for a cosine value from the angle between two vectors. The entities in classes are converted to the vectors by a model to calculate the angle between them. The value of cosine similarity is between 0 and 1. If the value is closer to 1, it means entities are more similar to each other. And if the value is closer to 0, it means the similarity between entities is less. In the following, the formula shows the semantic similarity between sets of entities ($e\textsubscript{1i}$ and $e\textsubscript{2j}$): 

\begin{equation}
sem\_sim(e\textsubscript{1i},e\textsubscript{2j}) = cos(\theta) = \dfrac{e\textsubscript{1i}.e\textsubscript{2j}}{|e\textsubscript{1i}||e\textsubscript{2j}|}   \end{equation}

There are many embedding models can be used to measure semantic similarity and relatedness by the cosine between the concepts’ embedding vectors. Some of these methods focus on the terms called word embedding, and some focus on the relations known as graph embedding.
Word2Vec presented by \cite{wordvec} as a word embedding model is used in this work to generate concept sentence embeddings based on the terms. Word2Vec model transforms words into low-dimensional word embeddings; it resorts to small neural networks to calculate these word embeddings based on contextual knowledge encoded in public background knowledge bases. 
In addition, to compute word embeddings, the Word2Vec model calculates the cosine of the angle between these low-dimensional vectors that represent these embeddings. Word2Vec model resorts to the cosine similarity as a measure to find similar words. In order to find similar words, a trained Word2Vec model based on gensim library\footnote{\url{https://radimrehurek.com/gensim/models/word2vec.html}} can be used. The main precondition of word embedding is that words with similar meaning should have a similar representation. There are many entities and relations that are semantically similar, but they are represented differently in the knowledge graph. Thus, considering context by additional training data is an important task to generate contextualized word embedding.

In the running example, Word2Vec is utilized to determine that the properties \textit{country}, \textit{nationality}, and \textit{birthPlace} are related and similar. Similarity values can guide the summary of properties in RDF graphs and allow for the transformation of SPARQL queries. \autoref{Simplequery} illustrates a transformed SPARQL query. Albeit simpler, the transformed SPARQL query can retrieve the same results as the original query, but in less time. Meanwhile, there are a number of applications where data are represented in the form of graphs, which required graph embedding.

RDF2Vec presented by \cite{RDF2Vec} as a graph embedding model is applied to learn the context of the relations. In the case of RDF knowledge graphs, entities and relations between entities are considered instead of word sequences. First, the graph data is converted into sequences of entities; it can be considered as sentences using two different approaches, i.e., graph walks and Weisfeiler-Lehman (WL) subtree RDF graph kernels. Using those sentences, RDF2Vec trains the same neural language models to represent each entity in the RDF graph as a vector of numerical values in a latent feature space.

Built on existing results on graph embeddings and summarization, we propose two approaches for summarizing RDF graphs. 
The first called Grouping Based Summarization (GBS) approach; it summarizes the RDF graph based on grouping subjects with the same predicates and objects.
The second one, optimized for the first one, called Query Based Summarization (QBS) considers only the part of the RDF graph which is related to the SPARQL query.
In the next, GBS and QBS are defined in detail. 

\subsection{A naive approach for summarizing RDF graphs}
\label{sub:twophase}
A Grouping Based Summarization (GBS) approach able to reduce size of RDF graph represents our naive method. 
GBS works into phases: \textit{Offline Phase} and \textit{Online Phase}. In \autoref{fig:sg-architec} both phases are shown. The input of the offline phase is an RDF dataset that has been loaded as an RDF graph, and the output is a semantic summary graph. For the online phase, the input is a generated summary graph from the offline phase, a SPARQL query, and a semantic similarity metric and the output is a transformation of the query to the one with fewer triple patterns with final answers. 

\begin{figure*}[t]
    \centering
    \includegraphics[width=1\columnwidth]{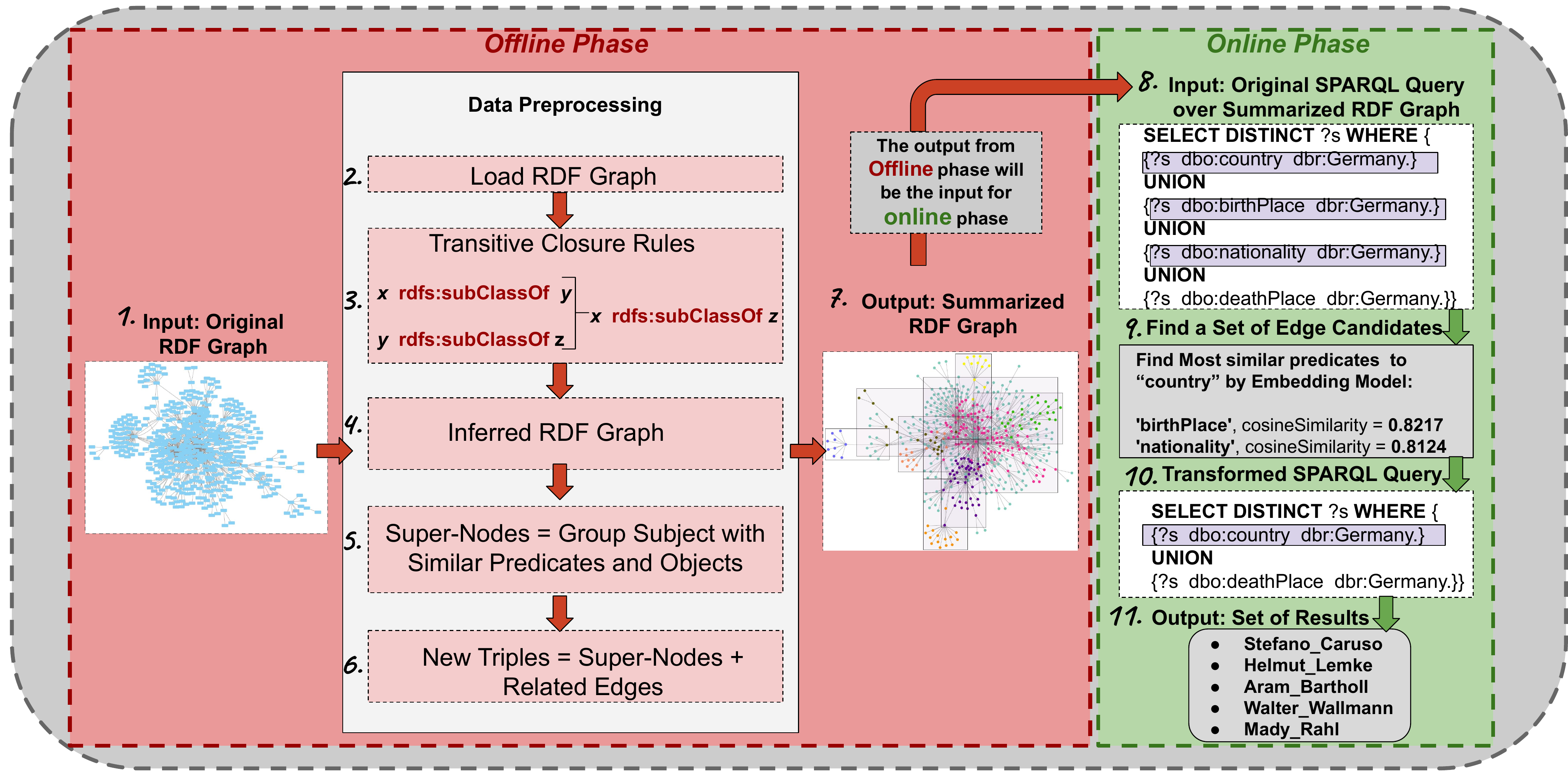}
    \caption{Proposed Summary Graph Architecture for the Grouping Based Summarization (GBS) approach in two phases (The Naive Approach).}
    \label{fig:sg-architec}
\end{figure*}

\begin{algorithm}
% \begin{algorithmic}
\SetAlgoLined
\SetKwInOut{Input}{Input}
\SetKwInOut{Output}{Output}
\Input{RDF dataset has been loaded as RDF graph (G)}
\Output{Semantic Summary Graph (G')}
\BlankLine
\nl G $\gets$ RDFGraphLoader.loadFromDisk(spark, input, parallelism) \;

\nl inferredGraph $\gets$ TransitiveReasoner.apply(G) \;

\nl SN $\gets$ RDD[(List[Subjects])] from inferredGraph where Subjects have similar Predicates and Objects \;  /*SN is a list of Super-Nodes*/

\nl newTriples $\gets$ Triple.create(SN, Predicate, Object) \;

\nl Buffer $\gets$ new ArrayBuffer(triple.length) \;

\nl \ForEach{(SN, Predicate, Object) $\in$ inferredGraph}{
\nl    \If{newTriples not exists in Buffer}{
\nl Buffer += newTriples \;
}
}

\nl G' $\gets$ Buffer \;

\nl \Return G'
\caption{The Grouping Based Summarization (GBS) Algorithm, \textit{Offline phase-Summarizing Graph}}\label{algo1}

% \end{algorithmic}
\end{algorithm}
In the offline phase presented in Algorithm~\ref{algo1}, data should be preprocessed. All edges and vertices are read from the RDF graph. As discussed before, the aim is to provide a semantic summary graph without losing the needed information. After receiving a dataset as an input in \textit{(step 1)}, the RDF graph is loaded \textit{(step 2)}. Then, the RDF graph is expanded to find new relations which are already not available in the original RDF graph, but they have similar semantic meanings. This extension is performed by computing the transitive closure of the properties in the RDF graph. 

Inference layer is used in order to extract new knowledge. For inferring the new facts from the current knowledge bases, inference rules are applied. Transitive Closure (TC) is one of the inference rules which is exerted in this work to infer more facts. It is deployed in the graph in order to find more facts which not exist in the original one.

For example, in the given original RDF graph, there is no relation between entities \textit{Germany} and \textit{Country}. As seen in \autoref{fig:tc}, by applying TC inference rule, the new triple $\langle$\textit{Germany}, \textit{type}, \textit{Country}$\rangle$ has been inferred out of existing triples $\langle$\textit{Germany}, \textit{type}, \textit{EuropeanCountry}$\rangle$ and 
$\langle$\textit{EuropeanCountry}, \textit{subClassOf}, \textit{Country}$\rangle$. Since the original RDF graph has expanded by applying inference rules with more triples, queries over \textit{Country} and \textit{Germany} can be equally answered.

\begin{figure*}[t]
    \centering
    \includegraphics[width=1\columnwidth]{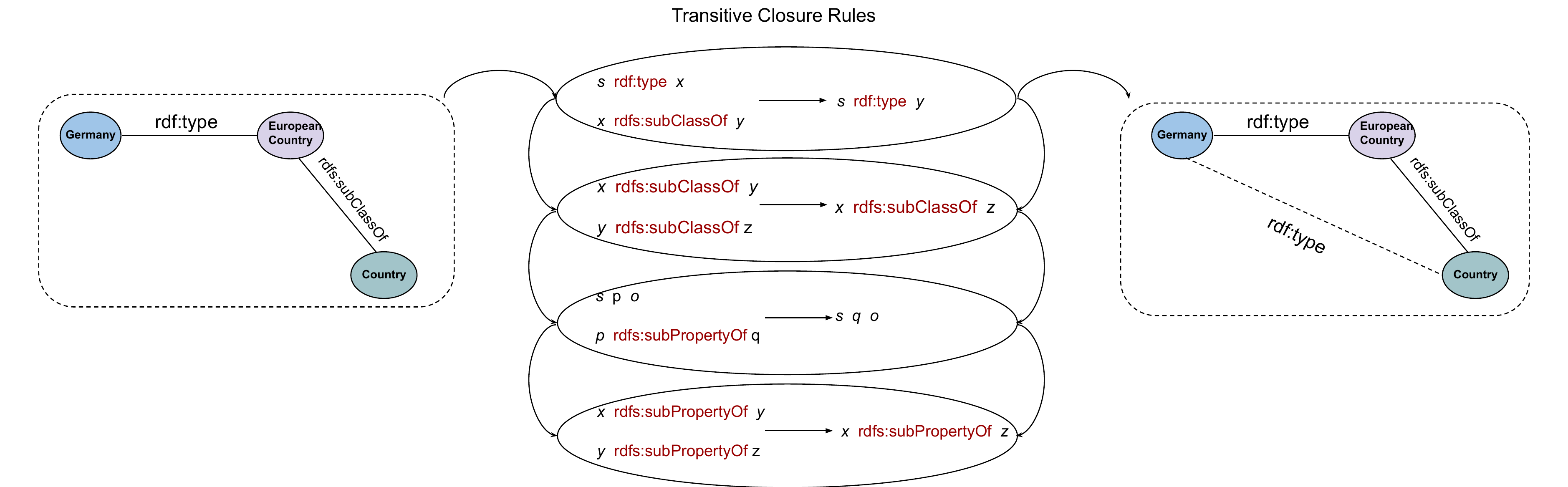}
    \caption{Example of Transitive Closure (TC) rule in Inference Layer}
    \label{fig:tc}
\end{figure*}

The Transitive Closure (TC) inference rule is deployed in the graph in order to find more facts \textit{(step 3)} and 
for expanding the RDF graph. In \textit{(step 4)}, the inferred RDF graph is generated with the new properties. 
Subjects with similar predicates and objects need to be identified to find Super-Nodes (SN). In order to store the results from massive data, in-memory Spark Resilient Distributed Dataset (RDD) is used \textit{(step 5)}.

The Resilient Distributed Dataset (RDD) is the core of Apache Spark~\footnote{\url{https://spark.apache.org}}. As it comes from the name, RDD is a resilient, distributed, and immutable collection of data that are partitioned over a cluster of machines. In Spark RDD, a cluster of workers is connected to a driver or master node. A master node will take care of work execution while worker nodes execute the jobs which are split and then distributed to them. 
In Sparklify, RDF graphs are stored and modeled based on Spark RDD through fast processing for efficient evaluation of SPARQL queries over distributed RDF datasets. 

\begin{figure*}[t!]
    \centering
     \vspace{0pt}\subfloat[Original RDF graph]{
      \includegraphics[width=.38\textwidth]{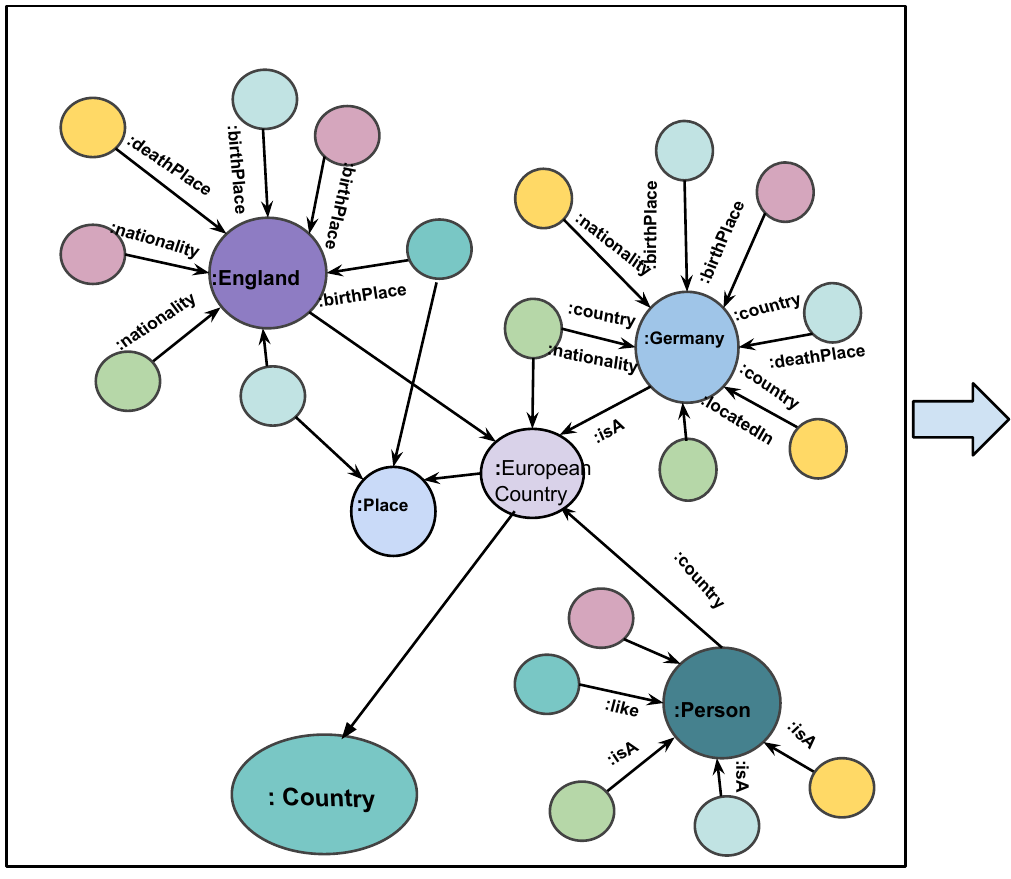}
      \label{fig:GBSexampleOff-a}}
      \vspace{0pt}\subfloat[Inferred graph by \newline finding new facts]{
      \includegraphics[width=.4\textwidth]{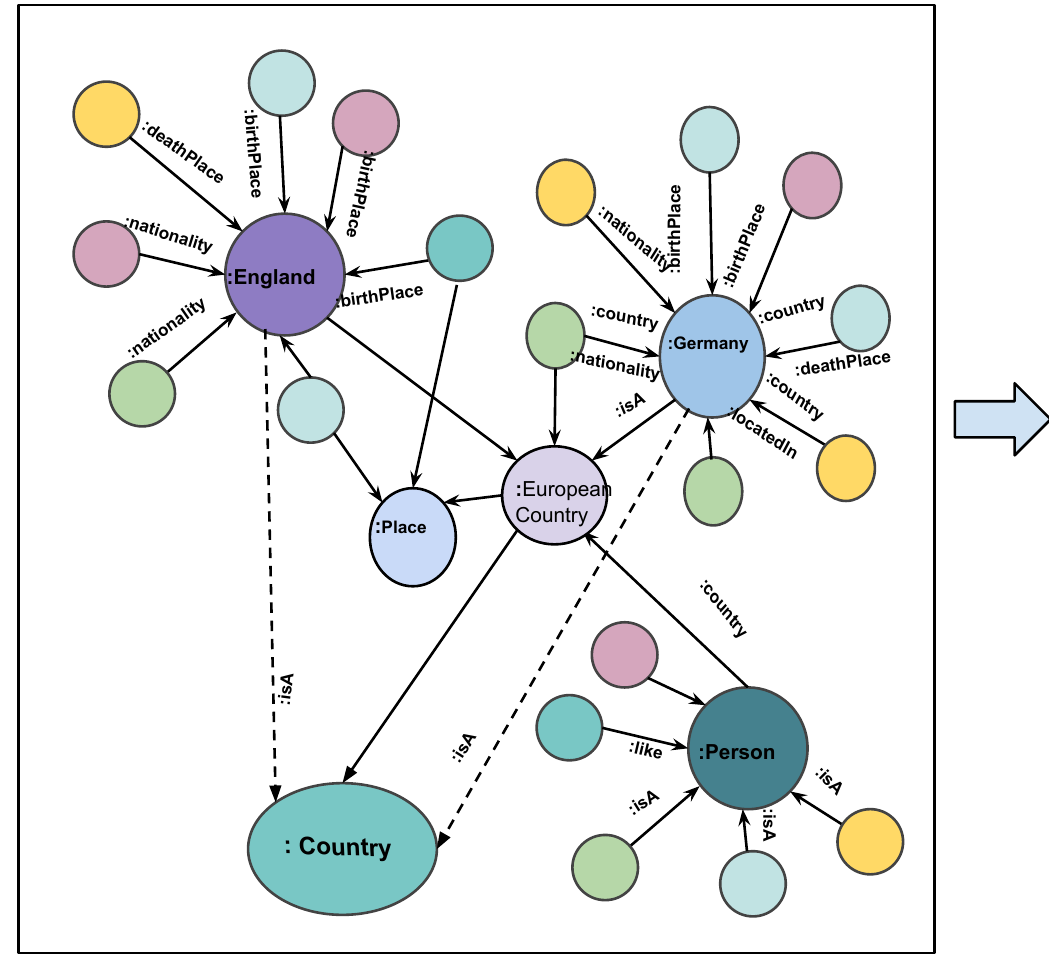}
      \label{fig:GBSexampleOff-b}}
      \hspace{0pt}\subfloat[Group subjects with the same \newline predicate and object using RDD Spark]{
      \includegraphics[width=.45\textwidth]{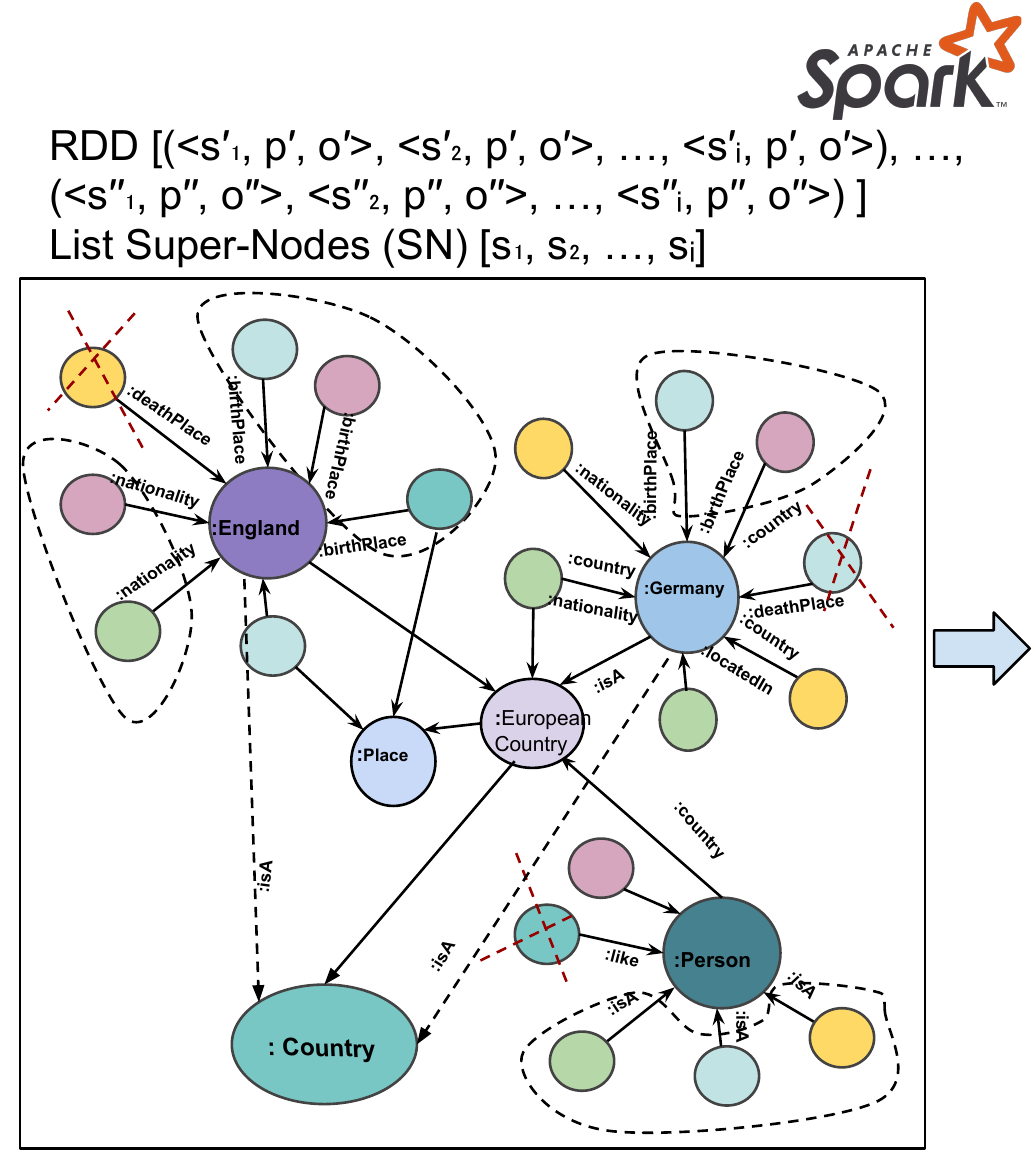}
      \label{fig:GBSexampleOff-c}}
      \vspace{0pt}\subfloat[Summarized RDF graph]{
      \includegraphics[width=.35\textwidth]{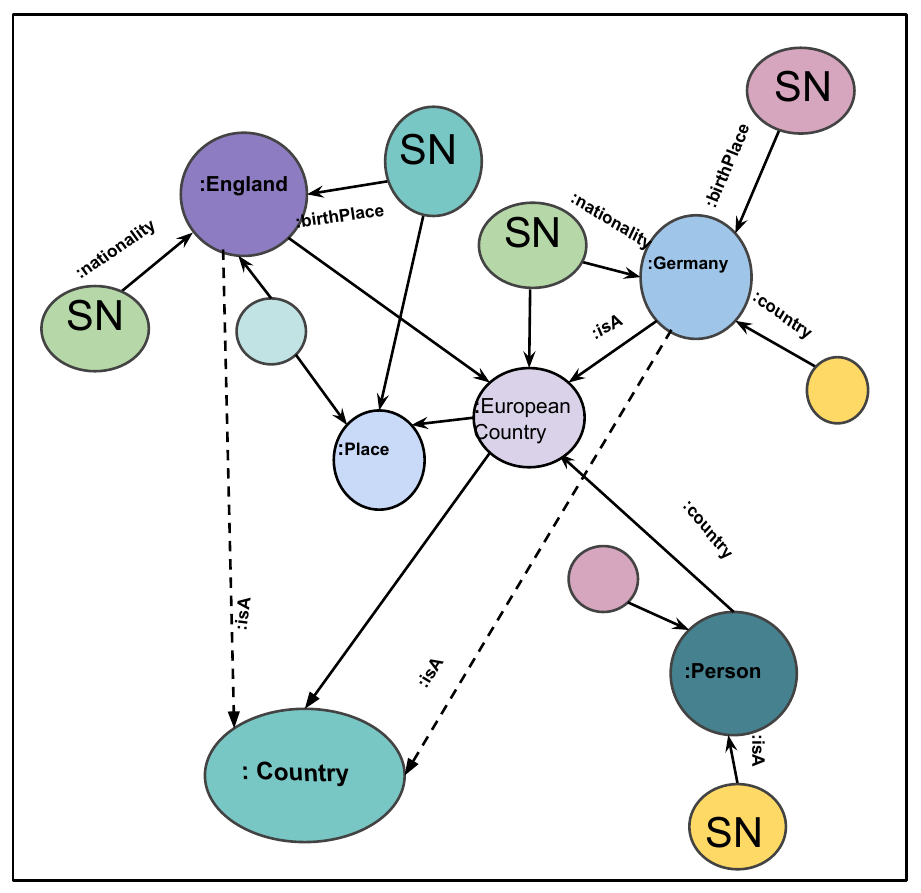}
      \label{fig:GBSexampleOff-d}}
    \caption{{\bf Example of Grouping Based Summarization (GBS) approach (offline phase).} (a) A part of original RDF graph; (b) Graph expanded using the Transitive Closure rule to find new facts which do not exist in the original RDF graph, but they are semantically true; (c) Subjects with the same predicate and object are grouped to find Super-Nodes (SN) to create new triples using RDD Spark, also the nodes that do not share the same predicate and object will be removed; (d) Generate a summary RDF graph with fewer nodes and edges. }
    \label{fig:GBSexample-off}
\end{figure*}

All subjects which have common predicates and objects are grouped. RDDs of (key,value) pairs are used like $(Predicate, Object, List(Subjects))$; this resembles the technique proposed by \cite{consens2015s+}; it also finds the same nodes by clustering entities with the same type. Pair RDDs display operations such as $reduceByKey()$ and $groupByKey()$ for combining and grouping values with the same key. Each of these RDDs of pairs can be considered as a Sub-Graph. Also, subjects in the group list are considered as Super-Nodes (SN). 
\autoref{fig:GBSexample-off} shows with a simple example how a summary RDF graph has been generated by GBS algorithm in offline phase.
For example, all people who were born in Germany can be grouped and consider as a Super-Node (SN). 
New triples are created out of these Super-Nodes with their related edges and added to the buffer\textit{(step 6)}. So, a summarized RDF graph is generated from the original one. 
In general, the total number of edges and vertices of this summary graph is less than the original RDF graph. Therefore, the summarized graph is created \textit{(step 7)}.

After generating a summary RDF graph, the aim is to have a complete set of answers corresponding to a reduced number of triple patterns in the query. 
In the online phase, a multi triple pattern query is processed in \textit{(step 8)} to find a set of edge candidates based on the embedding model such as Word2Vec \textit{(step 9)}. As explained in \autoref{sec:approach}, cosine similarity as a measure is used to find similar predicates by their distance from each other. Also, a trained model based on gensim library has been used. In our model, semantically similar predicates tend to lie close to each other. For example, the cosine similarity value between a given predicate $country$ and predicates $nationality$, $birthPlace$, and $deathPlace$ is $0.8217$, $0.8124$, and $0.3672$, respectively. In \textit{(step 9)}, the edges which have a higher similarity value than a given threshold (> $0.5$) are selected. Therefore, predicate $deathPlace$ cannot be considered as a similar predicate to $country$. The similar edges are considered as strong relations between vertices and are called Super-Edges.
As seen in Algorithm~\ref{algo2}, the Super-Edges discovered are used in transforming the query to the simple one to find complete results in \textit{(step 10)}. 
Word embedding techniques consider similarity between edges based on their distance. Thus, in summarized RDF graphs where the size of graph is smaller and predicates are closer to each other, there is a possibility that a founded predicate is similar to the others, not only in terms of distance but also from the semantic point of view. A simple example in \autoref{fig:GBSexample-on} demonstrates how the algorithm of GBS approach works in online phase.

\begin{algorithm}%[H]
\SetAlgoLined
\SetKwInOut{Input}{Input}
\SetKwInOut{Output}{Output}
\Input{SPARQL Query (Q); Summary Graph (G') from Algorithm~\ref{algo1}; List of Vocabulary (V); Semantic Similarity Metric}
\Output{Transformed SPARQL Query (Q'); List of result}
\BlankLine
\nl initialize training model \;
\nl V $\gets$ list\_of\_vocabulary \;
\nl \ForEach{vocabulary $\in$ V}{
\nl model $\gets$ trained model \; }

\nl Q.Predicates $\gets$ set of predicates extracted from Q \;

\nl \ForAll{p $\in$ Q.Predicates}{
\nl similar\_set\_of\_P $\gets$ predicates q in G' with cosine-similarity(q,P) > 0.5 \;
\nl Q.replaceBy(P, representativeOf(synonym\_set\_of\_P)) \;
}
\nl Q' $\gets$ Q \;
\nl result $\gets$ G'.sparql(Q') \; 
\nl \Return result

\caption{The Grouping Based Summarization (GBS) Algorithm, \textit{Online Phase-Query Rewriting}}\label{algo2}
\end{algorithm}

\begin{figure*}[t]
    \centering
      \vspace{0pt}\subfloat[Summary graph, a SPARQL \newline query, and vocabulary of graph \newline predicates]{
      \includegraphics[width=.47\textwidth]{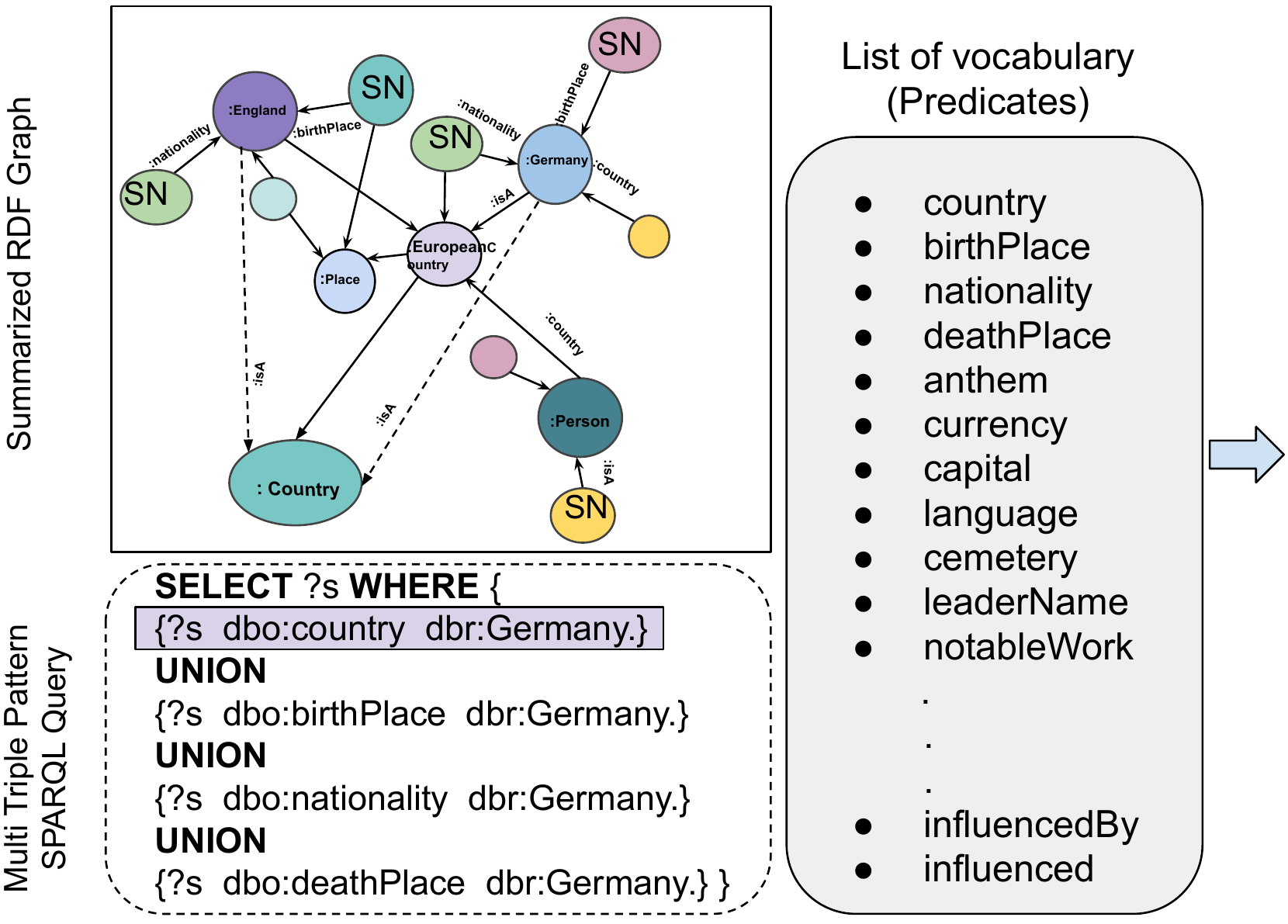}
      \label{fig:GBSexampleOn-a}}
      \vspace{0pt}\subfloat[Find similar predicates  \newline to transform the query]{
      \includegraphics[width=.45\textwidth]{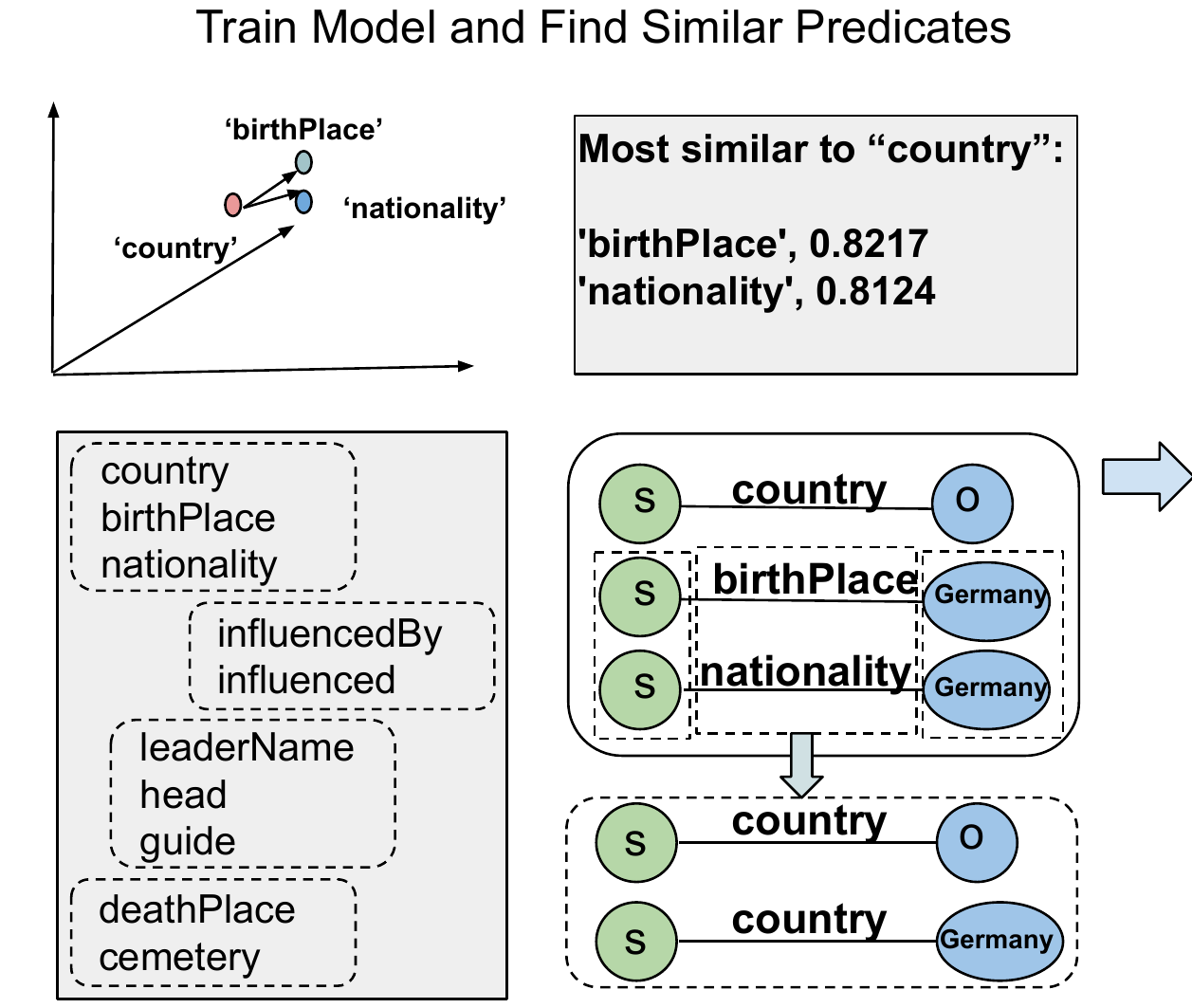}
      \label{fig:GBSexampleOn-b}}
      \hspace{0pt}\subfloat[A transformed query and answers]{
      \includegraphics[width=.6\textwidth]{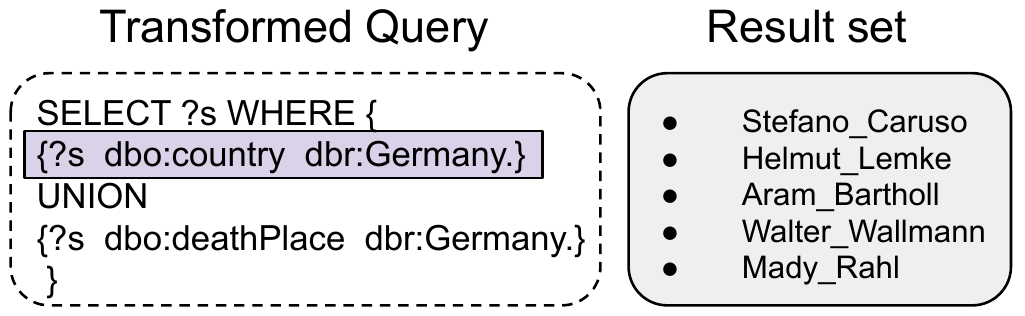}
      \label{fig:GBSexampleOn-c}}
    \caption{{\bf Example of Grouping Based Summarization (GBS) approach (online phase).} (a) The summary RDF graph generated from offline phase with a multi triple pattern SPARQL query and a list of vocabulary before training the model; (b) Apply both word embedding model and graph embedding model to find most similar predicates to transform multi triple pattern queries to simple ones by considering the most similar predicate as a Super-Edge; (c) Transformed query with results. }
    \label{fig:GBSexample-on}
\end{figure*}

In \textit{(step 11)}, final answers are generated by getting help from a query engine, e.g., Sparklify, over the summary RDF graph. After retrieving answers, it is observed that some required information is lost. In \autoref{sec:evaluation}, the results of evaluation show that querying over summarized RDF graph in GBS approach for large RDF graphs returns less number of answers compared with querying over original RDF graph. In order to avoid the problem of losing information during query processing, we propose  the Query Based Summarization (QBS) approach. 

\begin{algorithm}[t!]
\SetAlgoLined
\SetKwInOut{Input}{Input}
\SetKwInOut{Output}{Output}
\Input{RDF dataset loaded as RDF graph (G); SPARQL Query (Q); List of Vocabulary (V); Semantic Similarity Metric}
\Output{Transformed SPARQL Query (Q''); Summary Graph (G''); List of result}
\BlankLine
\nl G $\gets$ RDFGraphLoader.loadFromDisk(spark, input, parallelism) \;

\nl Super-Predicates $\gets$ Q.getSetOfPredicates \;

\nl Super-Objects $\gets$ Q.getSetOfObjects \;

\nl g $\gets$ triples includes Super-Predicates or Super-Objects \; /*g is a Sub-Graph*/

\nl initialize training model \;
\nl V $\gets$ list\_of\_vocabulary \;
\nl \ForEach{vocabulary $\in$ V}{
\nl model $\gets$ trained model \; }

\nl Q.Predicates $\gets$ set of predicates extracted from Q \;

\nl \ForAll{p $\in$ Q.Predicates}{
    \nl similar\_set\_of\_P $\gets$ predicates q in g with cosine-similarity(q,P) > 0.5 \;
    \nl Q.replaceBy(P, representativeOf(synonym\_set\_of\_P)) \;
    }

\nl Q'' $\gets$ Q \;
\nl tsp $\gets$ G.getSetOfPredicates.getURI.contains(similar\_set\_of\_P) \; /*tsp is a set of triples contain similar predicates*/ 

\nl Super-Subjects $\gets$  tsp.getSetOfSubjects \;

\nl newTriples $\gets$ Triple.create(Super-Subjects, Super-Predicates, Super-Objects) \;

\nl Buffer $\gets$ new ArrayBuffer(triple.length) \;

\nl \ForEach{(Super-Subject, Super-Predicate, Super-Object) $\in$ G}{
\nl    \If{newTriples not exists in Buffer}{
\nl Buffer += newTriples \;
}
}

\nl G'' $\gets$  Buffer.union(g) \;

\nl result $\gets$ G''.sparql(Q'') \; 

\nl \Return result
\caption{The Query Based Summarization (QBS) Algorithm}\label{algo3}
\end{algorithm}

\begin{figure*}[t]
    \centering
    \includegraphics[width=1\columnwidth]{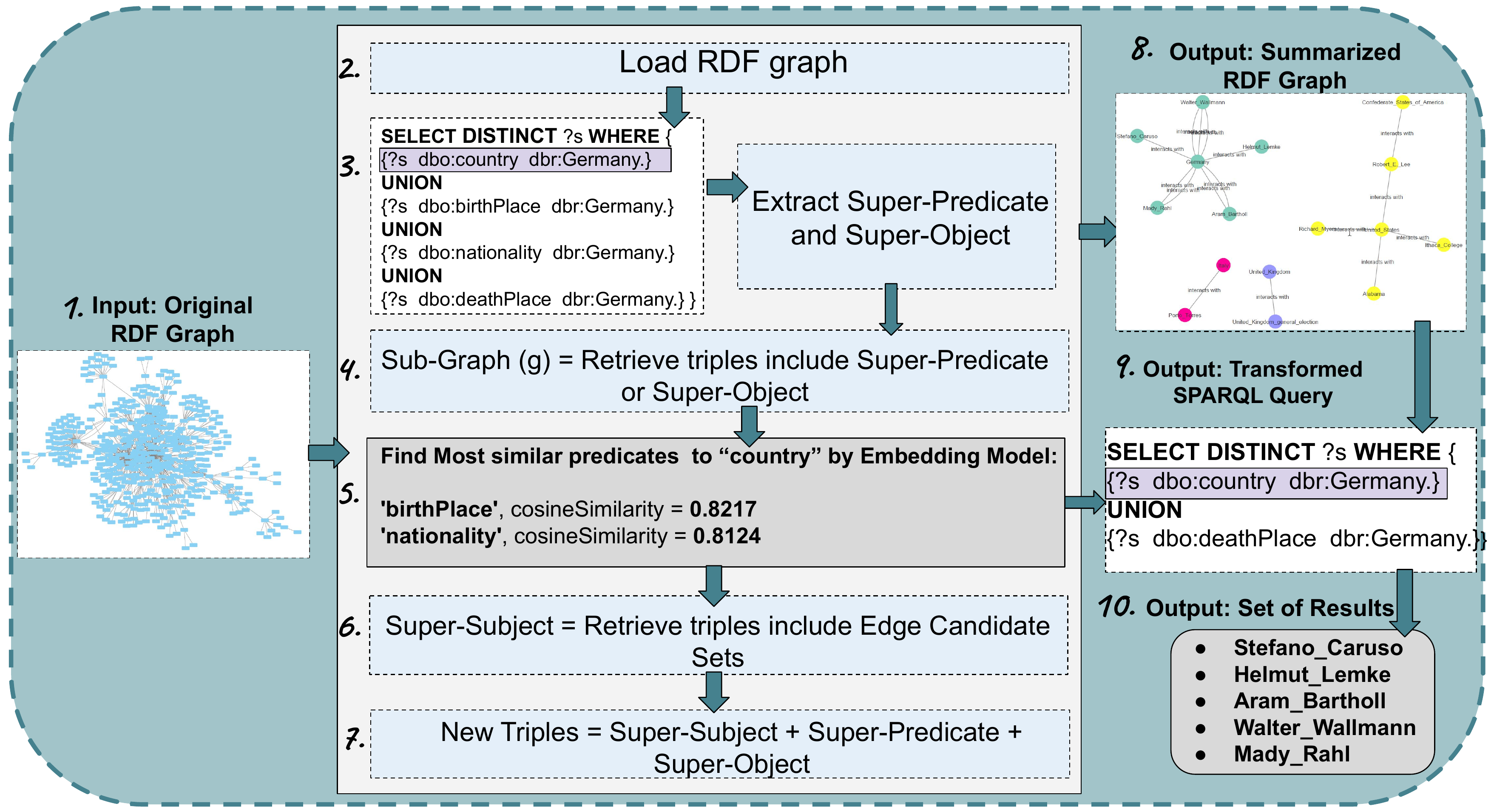}
    \caption{Summary Graph Architecture for the Query Based Summarization (QBS) approach (The Optimized Approach).}
    \label{fig:sg-architec2}
\end{figure*}

\subsection{An optimized approach for summarizing RDF graphs}

The optimized approach in order to reduce size of RDF graphs is Query Based Summarization (QBS) approach. QBS guides the summarization process based on an input query. 
Thus, instead of considering a whole RDF graph to summarize, only the part of the original RDF graph which is related to the user query will be considered. This idea is inspired by \cite{ASSG}. In this way, not only the execution time is reduced, but also all the necessary information is preserved. The input of this algorithm is an RDF dataset that has been loaded as an RDF graph, a SPARQL query, and a semantic similarity metric. The output is a transformation of the query to the simple one with fewer triple patterns and a semantic summary graph with final answers; it is presented in Algorithm~\ref{algo3}. After receiving an RDF graph as input and load it \textit{(step 1-2)}, the query of the user should be collected to extract predicates and objects related to the query as Super-Predicates and Super-Objects \textit{(step 3)}. In this approach, contrary to GBS, the inferred graph is not generated due to the time-consuming and considering small part of graph related to the query. Therefore, the RDF graph is not expanded like in GBS approach. 

After extracting the Super-Predicates and Super-Objects in \textit{(step 3)}, a Sub-Graph ($g$) consisting of triples with predicates equal to the Super-Predicates or with objects equal to the Super-Objects is generated \textit{(step 4)}. The architecture of Query Based Summarization (QBS) is illustrated in \autoref{fig:sg-architec2}.
In the next step, word and graph embedding models are applied to this Sub-Graph to find edge candidate sets. Unlike the GBS approach, the embedding models consider the Sub-Graph instead of the whole graph. By this method, only relevant predicates will be found as similar \textit{(step 5)}. Indeed, embedding models help that query is transformed into a simple SPARQL query. It also helps that Super-Subjects are found by extracting triples with the predicates equal to the edge candidate sets \textit{(step 6)}. By having Super-Subjects, Super-Predicates, and Super-Objects from the previous steps, new triples are created in \textit{(step 7)} and added to our Sub-Graph ($g$)  to generate the final graph as a summary graph \textit{(step 8)}. The Super-Edges discovered earlier are used in transforming the query to the simple one \textit{(step 9)}. 

Since QBS relies on the query of users, the generated summary graph contains all information related to the query. Therefore, querying the transformed query over the summary graph returns all possible answers. This is proved by Theorem~\ref{theo2}. By applying the simple query over this summary graph, the answers are retrieved \textit{(step 10)}.

\begin{theorem}
\label{theo2}
If $C\textsubscript{1}$ and $C\textsubscript{2}$ are classes in the RDF graph $G$, and $C\textsubscript{1}$ is the domain of $p\textsubscript{1}$ and $C\textsubscript{2}$ is the domain of $p\textsubscript{2}$. Also, $p\textsubscript{1}$ is similar to $p\textsubscript{2}$ according to a given \textbf{semantic similarity metric}. Let $G''$ be the compact representation of $G$ by QBS approach, where $p$ is the property used to represent $p\textsubscript{1}$ and $p\textsubscript{2}$ in $G''$. 
The following properties hold: 

\begin{enumerate}
\item The cardinality of $G''$ and cardinality of $G$ are the same.
\begin{equation}
\textit{cardinality($G''$) = cardinality($G$)}
\end{equation}
\item For SPARQL query $Q$, where $p\textsubscript{1}$ or $p\textsubscript{2}$ are used in the triple patterns of $Q$. $Q$ can be a conjunctive query or include the UNION or OPTIONAL operator. If $Q''$ is the transformation of $Q$ where $p\textsubscript{1}$ and $p\textsubscript{2}$ are replaced by $p$, the evaluation $Q'$ in $G'$ and 
the evaluation of $Q$ in $G$ are the same.
\end{enumerate}
\begin{equation}
[[Q]]G= [[Q'']]G''
\end{equation}
\end{theorem}

\begin{proof}
Consider the Query Based Summarization approach, which is dependent on the query of the user to summarize the RDF graph. All the properties and their relations related to the query appear in $G''$; it is the compact representation of $G$. Thus, $G''$ consists of all relations and entities related to the query $Q''$. Therefore, the properties $p\textsubscript{1}$ and $p\textsubscript{2}$ in $G$ which are similar based on the semantic similarity metric, and they belong to different classes are considered as predicate $p$ in $G''$. Hence, the cardinality of $G''$ is equal to the cardinality of $G$. Again, the query $Q$ with properties of $p\textsubscript{1}$ and $p\textsubscript{2}$ is rewritten to query $Q''$ with the property of $p$, and duplicated triple patterns are eliminated. So, query $Q$ over $G$ is the same as query $Q''$ over $G''$.
\qed
\end{proof}

\begin{figure*}[t!]
    \centering
     \vspace{0pt}\subfloat[RDF graph, a SPARQL query, and \newline vocabulary of graph predicates]{
      \includegraphics[width=.5\textwidth]{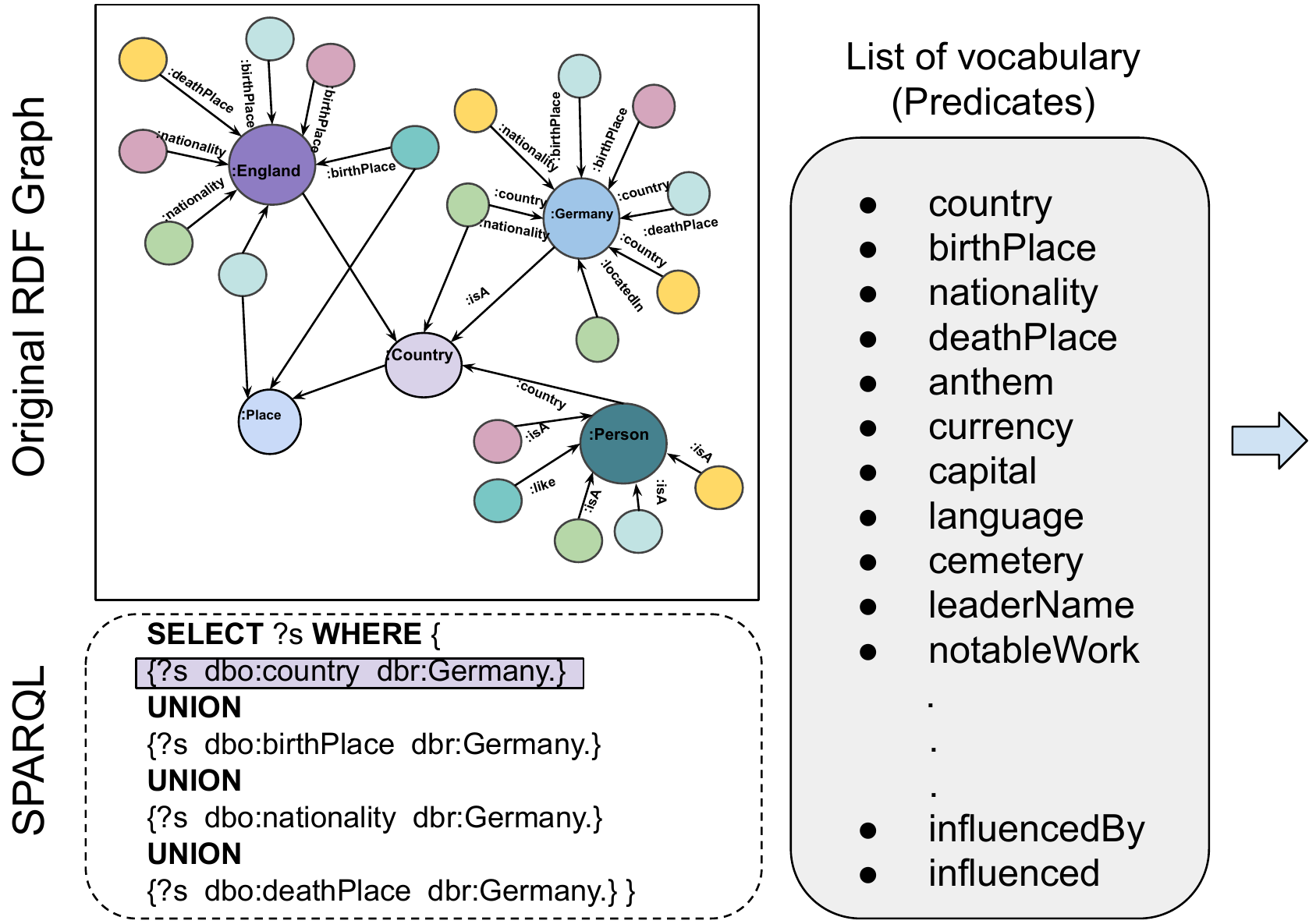}
      \label{fig:QBSexample-a}}
      \hspace{0pt}\subfloat[Sub-Graph ($g$) with selected predicate and object]{
      \includegraphics[width=.28\textwidth]{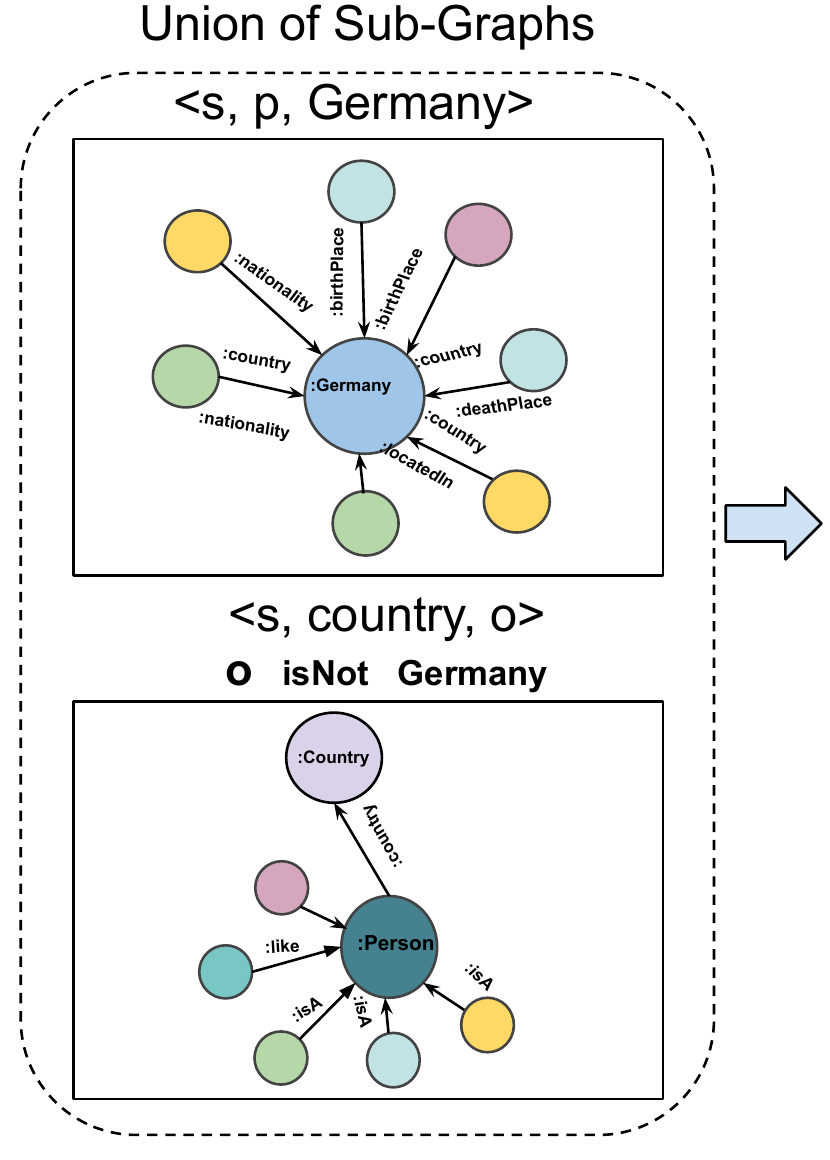}
      \label{fig:QBSexample-c}}
      \hspace{0pt}\subfloat[Add new triples with \newline similar predicates]{
      \includegraphics[width=.45\textwidth]{Fig7-b.pdf}
      \label{fig:QBSexample-d}}
      \vspace{0pt}\subfloat[Summarized RDF graph with \newline transformed query and a result set]{
      \includegraphics[width=.4\textwidth]{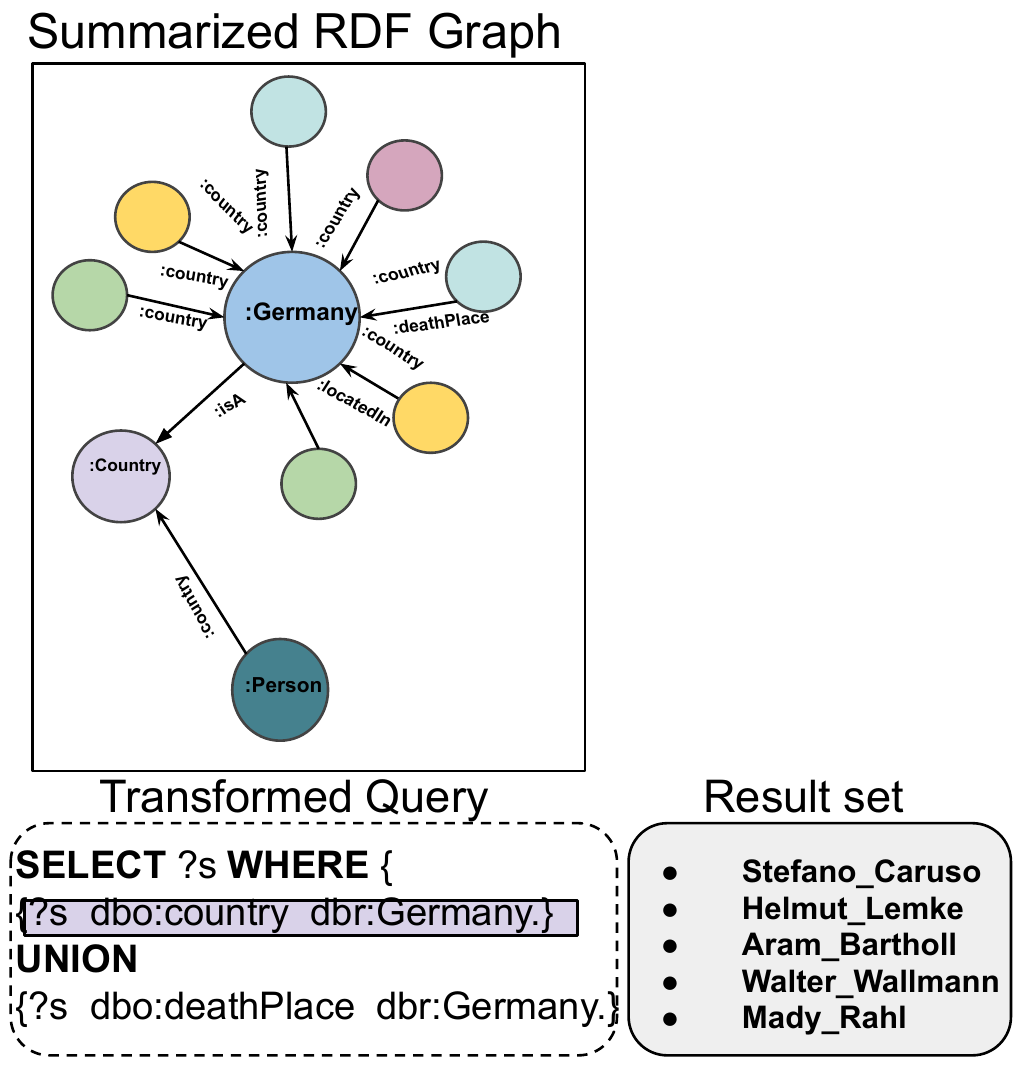}
      \label{fig:QBSexample-e}}
    \caption{{\bf Example of Query Based Summarization (QBS).} (a) The original RDF graph portion and a SPARQL query with predicate vocabulary before training the model; (b) Predicate and object of one triple pattern are considered as Super-Edge and Super-Object, respectively to extract the Sub-Graph ($g$); (c) The embedding models encode predicates as vectors to find the most similar predicates to the Super-Edge trained to create new triples; (d) Summary RDF graph with a transformed query and answers. }
    \label{fig:QBSexample}
\end{figure*}

\autoref{fig:QBSexample} illustrates the QBS algorithm with a simple example. A given query in Figure \autoref{fig:QBSexample-a} can also be considered as a multi triple pattern query without union operation. The number of answers by querying over the original RDF graph is equal to the number of answers by querying over the summarized graph. This has been proved in Theorem~\ref{theo2}.

\section{Empirical evaluation}
\label{sec:evaluation}
The effectiveness of the approaches described in \autoref{sec:approach} is analyzed based on the size of the RDF graph, the number of retrieved data, scalability, and query processing. We aim to answer the following research questions:
\begin{inparaenum}[\bf RQ1\upshape)]
     \item What is the impact of predicate relatedness by analyzing similarity measures in different size of datasets using diverse embedding models on finding complete answers?
    \item Are the proposed RDF summary graph techniques able to reduce the size of the original RDF graph by considering only required triples for querying?
    \item How is the effectiveness of the proposed summarized techniques on answer completeness and cardinality compared with the original RDF graph?
    \item How is the impact of the complexity of query on query processing and execution time in the proposed RDF summary graph compared with the original RDF graph?
\end{inparaenum}

\subsection{Data preparation and methods}
GBS and QBS have been evaluated in terms of compactness, completeness, and execution time. The summarization ratio is calculated by comparing the size of original RDF graphs with the size of summary graphs generated by two summarization approaches. The number of retrieved answers over original graphs is compared with the number of answers by querying over summarized graphs. Finally, the speed of query processing is computed by total time of returning answers by original and summary graphs using Sparklify.
One of the most important aspects of our evaluation is to collect and process the dataset to ensure that it is being worked by our approach. Therefore, the first step is preparing datasets to meet our conditions as an input. One condition is the dataset does not contain any literals. Since literal values and properties coming with literal values cannot be embedded by knowledge graph embeddings techniques. Then, RDF dataset should be loaded as an RDF graph. Later, the complex query is evaluated against this graph, and the answers are returned by Sparklify. The number of results and the running time are measured. For the proposed approaches, after loading the RDF dataset as an original RDF graph, the summarization technique is applied to generate a summary RDF graph. The number of answers retrieved by querying transformed queries over the summarized graph and query processing time are compared with the results produced over the RDF graph.

\subsection{Experimental setup}
We conduct an experimental study to assess the accuracy of our approach compared with the baseline. Our experimental configuration involves the datasets and queries used for our evaluation, as well as metrics and implementation.

\subsubsection{Datasets and queries} 
Four datasets with different sizes are applied to realize the effectiveness of the mentioned techniques. One is a small part of DBpedia dataset consists of \textit{2,047} triples wherein a bunch of companies, persons, and places with some of their information are stored. The second one is an Entity Summarization BenchMark (ESBM) with \textit{6,584} triples which are sample entities from two datasets, DBpedia and LinkedMDB (\cite{Hassanzadeh2009LinkedMD}) a popular movie database. The other selected datasets are Waterloo SPARQL Diversity Test Suite (WatDiv) with \textit{10,916,457} triples (WatDiv.10M) and with \textit{108,997,714} triples (WatDiv.100M) as medium and large datasets, respectively. Our evaluation is composed of 15 queries selected from QALD-3\footnote{\url{http://qald.aksw.org/index.php?x=task1&q=3}} for DBpedia dataset, from ESBM Benchmark v1.2 \footnote{\url{https://w3id.org/esbm/}} for ESBM, and from Query Generator (v0.6)\footnote{\url{https://dsg.uwaterloo.ca/watdiv/\#download}} for WatDiv.

\subsubsection{Metrics} 
Three main metrics in the evaluation of summary graphs are considered to answer the above research questions. \begin{inparaenum}[\bf a\upshape)]\item \textit{Compactness}: the size of summary RDF graph should be typically smaller than the size of original RDF graphs. The Summarization Ratio (SR) is a metric to show the value of compactness; it is calculated by the number of triples in the summary graph divide into the number of triples in the original RDF graph. \item \textit{Cardinality}: the number of answers returned by a query over the original RDF graph should be the same as over the summarized graph. \item \textit{Execution Time}: the running time for query processing over the summarized RDF graph should be less than over the original RDF graph.
\end{inparaenum}

\subsubsection{Implementation} 
The approaches are implemented in Scala 2.11.12 and Spark 2.2.0 over query engine Sparklify from the SANSA Stack framework. Sparklify executed over the original RDF graph is used as the baseline of our work. The proposed algorithms are compared with the original RDF graph based on Sparklify, \begin{inparaenum}[\bf i\upshape)]\item\textbf{Sparklify+GBS}; and \item\textbf{Sparklify+QBS}. \end{inparaenum} We evaluate our experiments on three servers with 256 cores and executor memory of 100 GB. 

\begin{figure*}[t!]
\centering  
\hspace{0pt}\subfloat[DBpedia]{
     \includegraphics[width=0.5\textwidth]{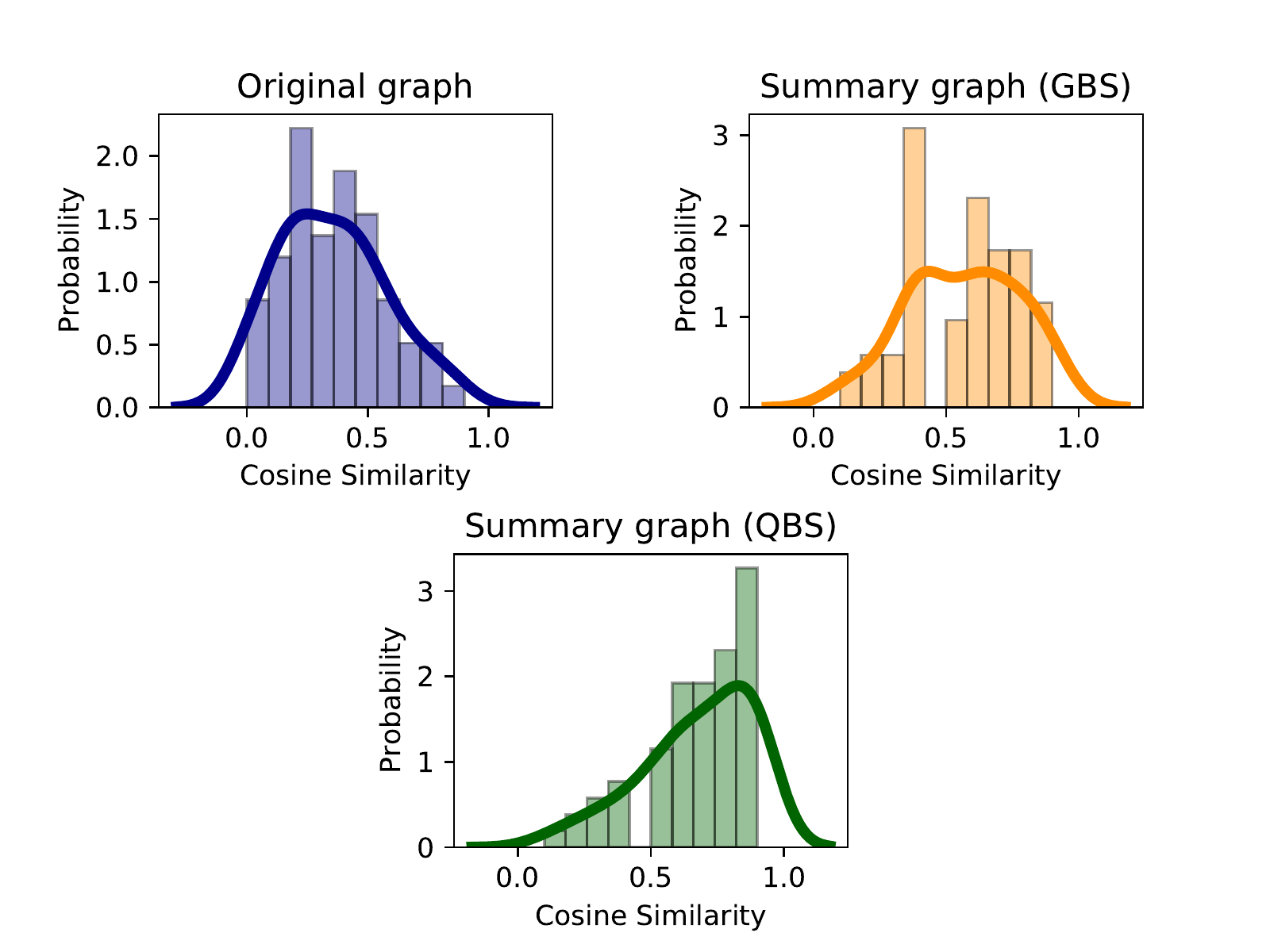}
     \label{fig:distplot1}}
     \vspace{0pt}\subfloat[ESBM]{
     \includegraphics[width=0.5\textwidth]{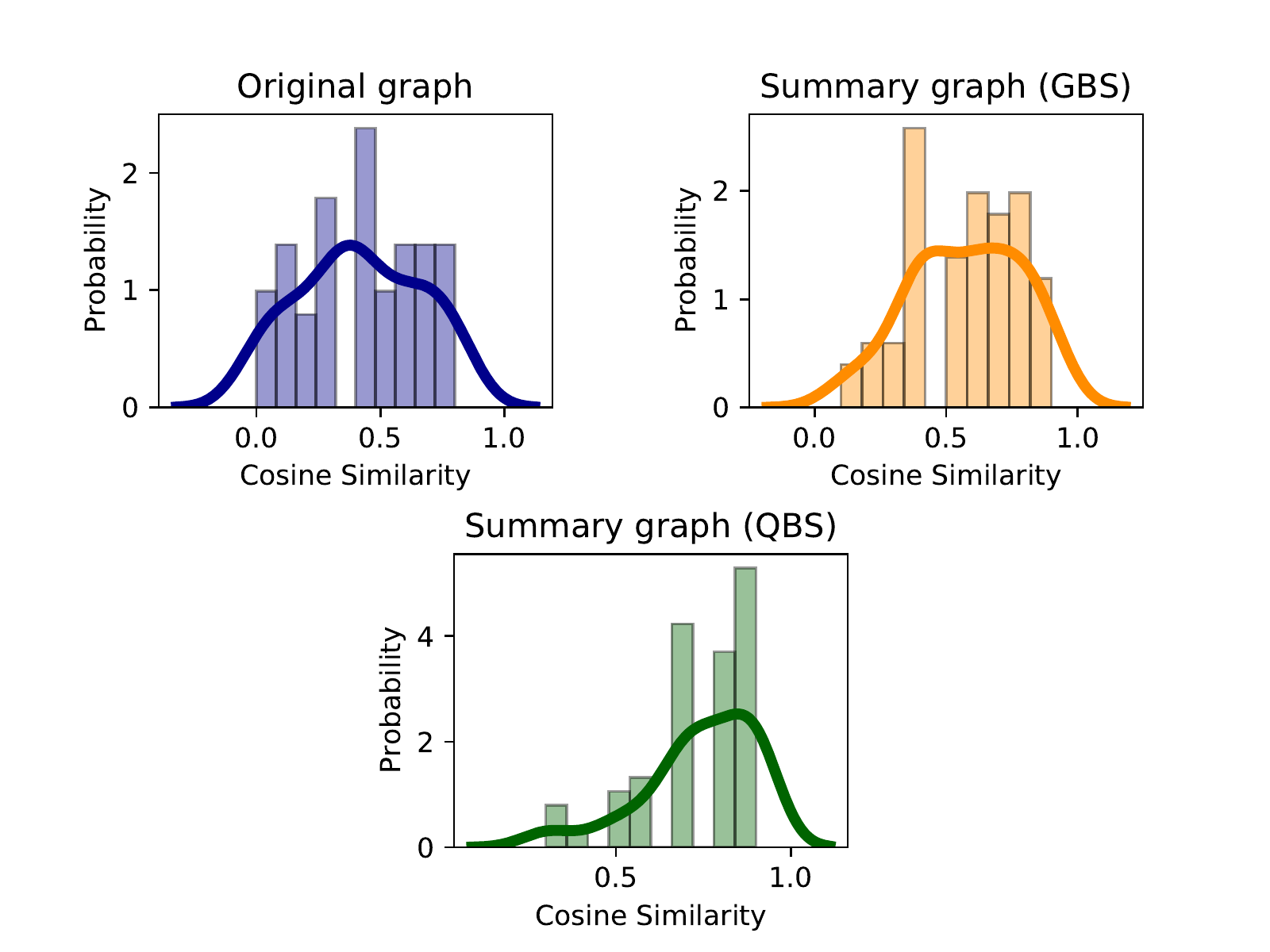}
     \label{fig:distplot2}}
\hspace{0pt}\subfloat[WatDiv.10M]{
     \includegraphics[width=0.5\textwidth]{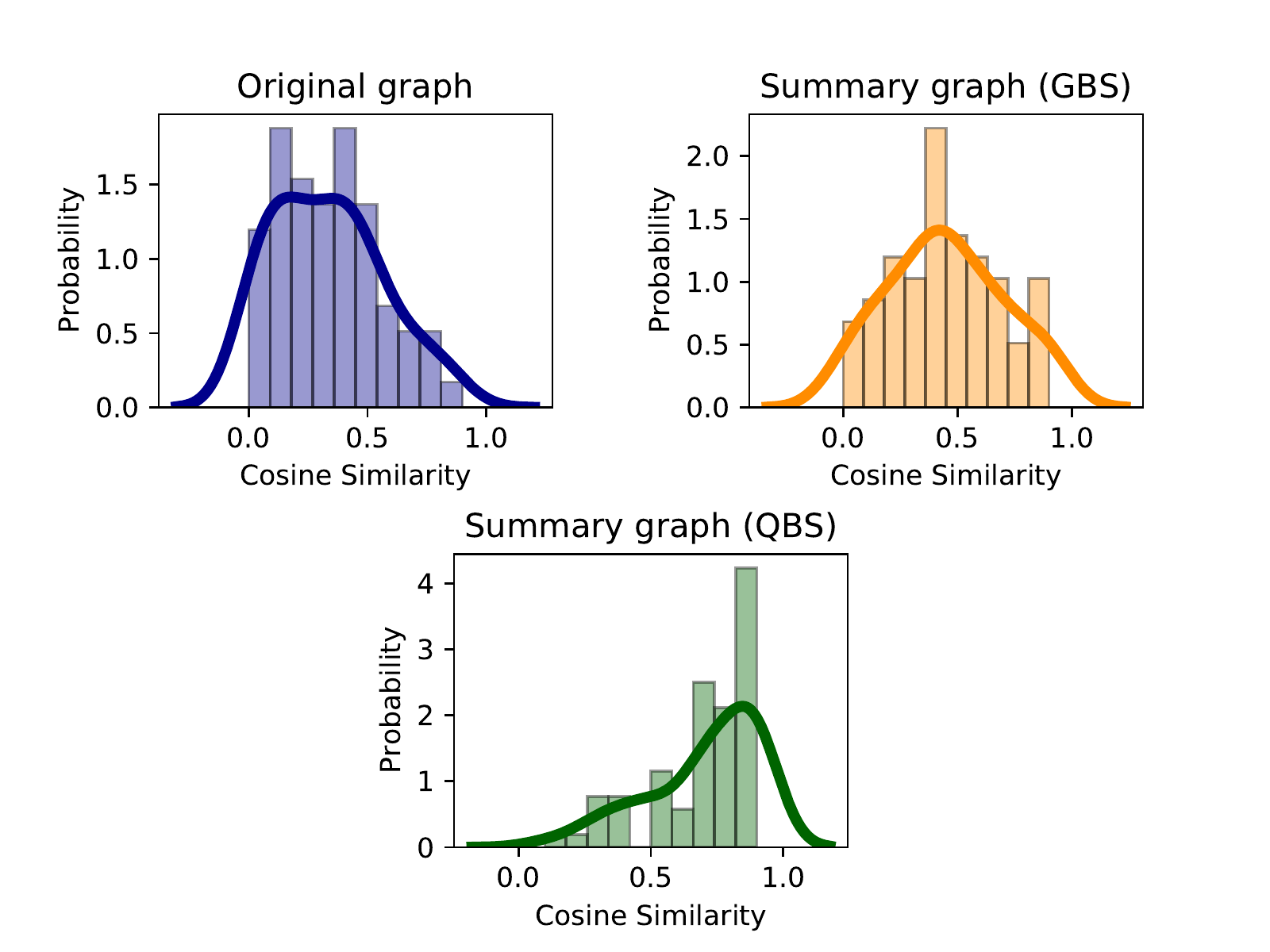}
     \label{fig:distplot3}}
     \vspace{0pt}\subfloat[WatDiv.100M]{
     \includegraphics[width=0.5\textwidth]{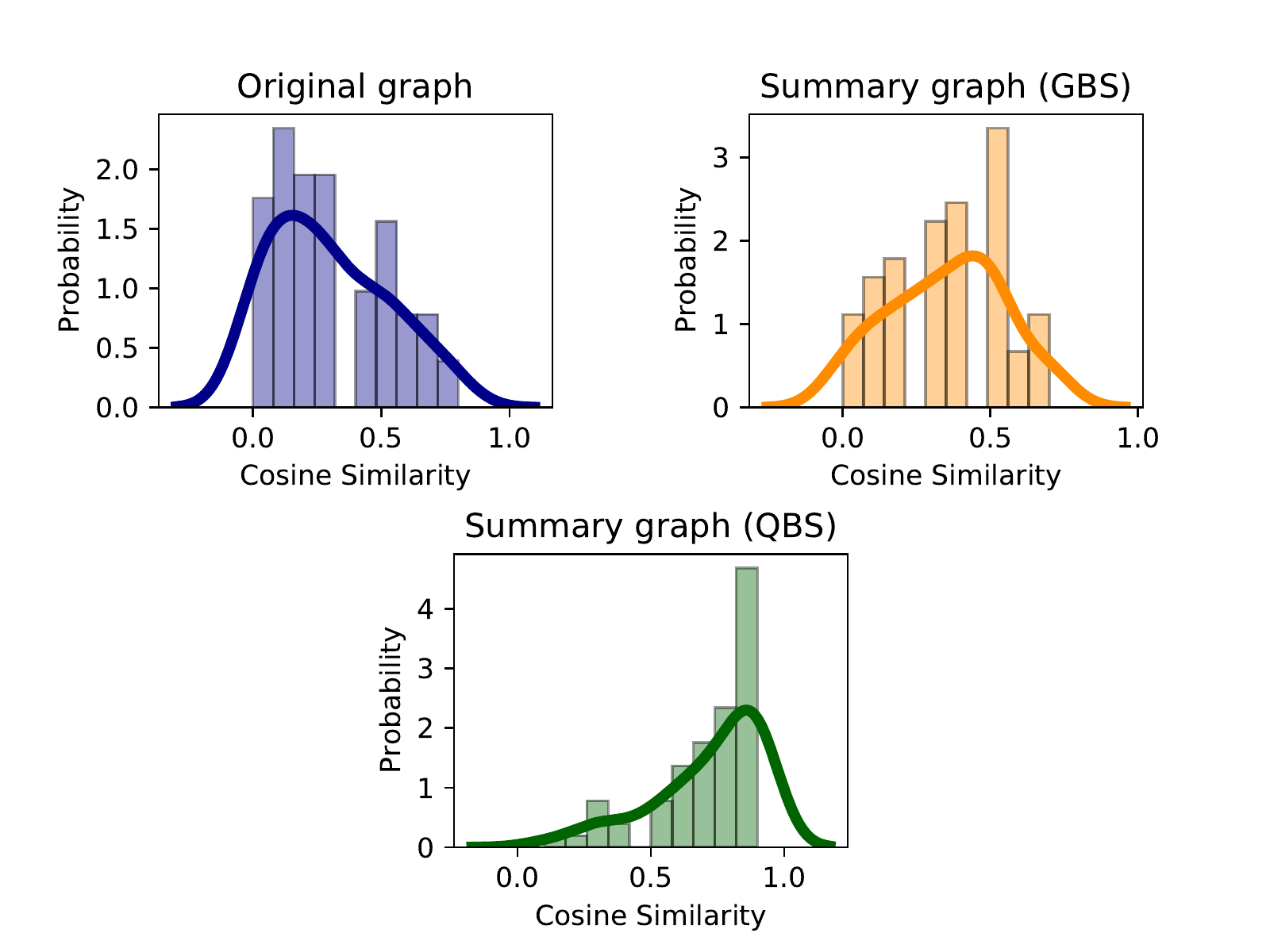}
     \label{fig:distplot4}}
\caption{Compare the distribution of cosine similarity for predicate relatedness in the original and the summarized RDF graph based on GBS and QBS approaches using word embedding model.}
 \label{fig:distribution}
\end{figure*}

\begin{figure*}[t!]
\centering  
\hspace{0pt}\subfloat[DBpedia]{
     \includegraphics[width=0.5\textwidth]{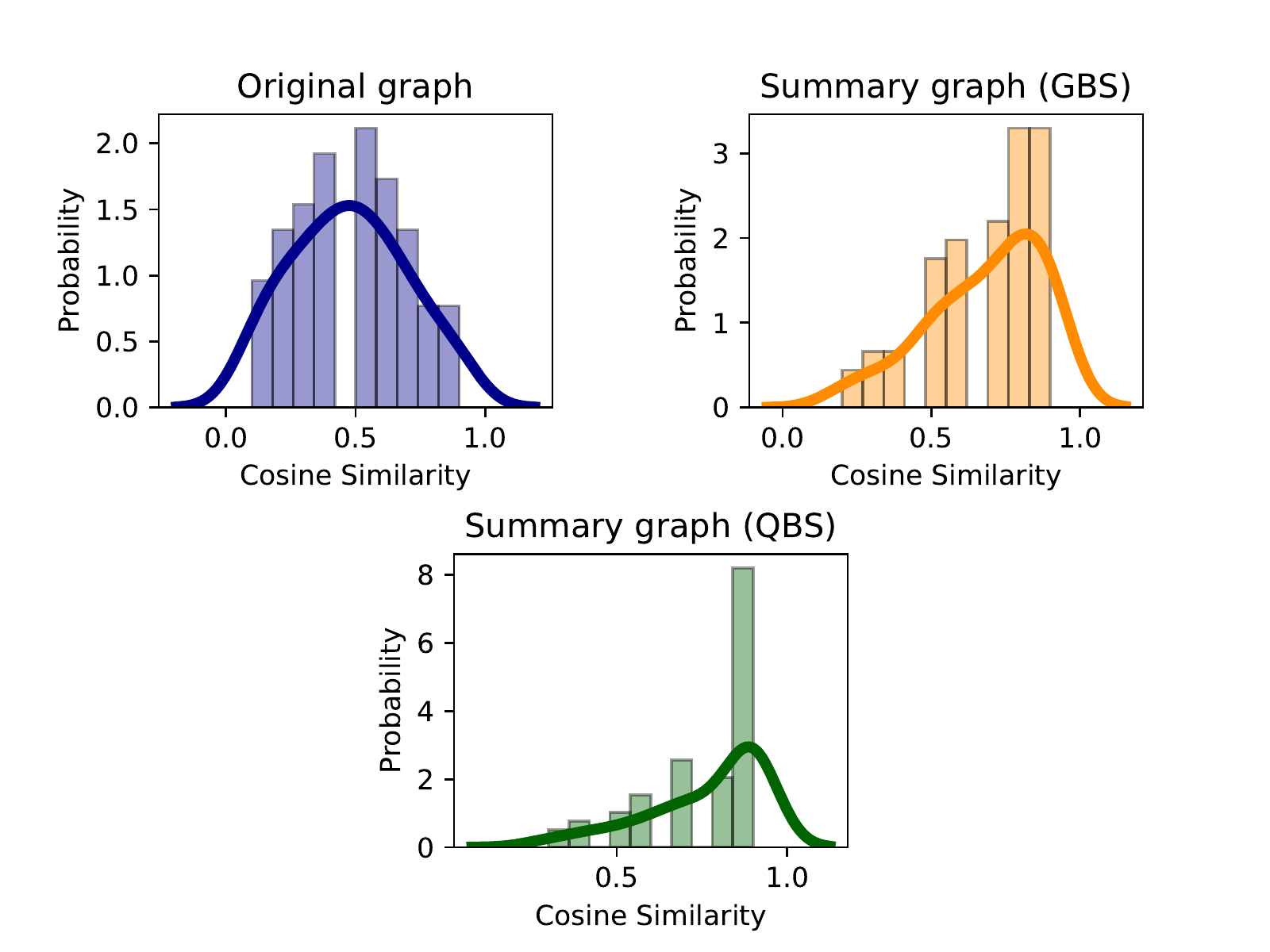}
     \label{fig:distplot1RDF2Vec}}
     \vspace{0pt}\subfloat[ESBM]{
     \includegraphics[width=0.5\textwidth]{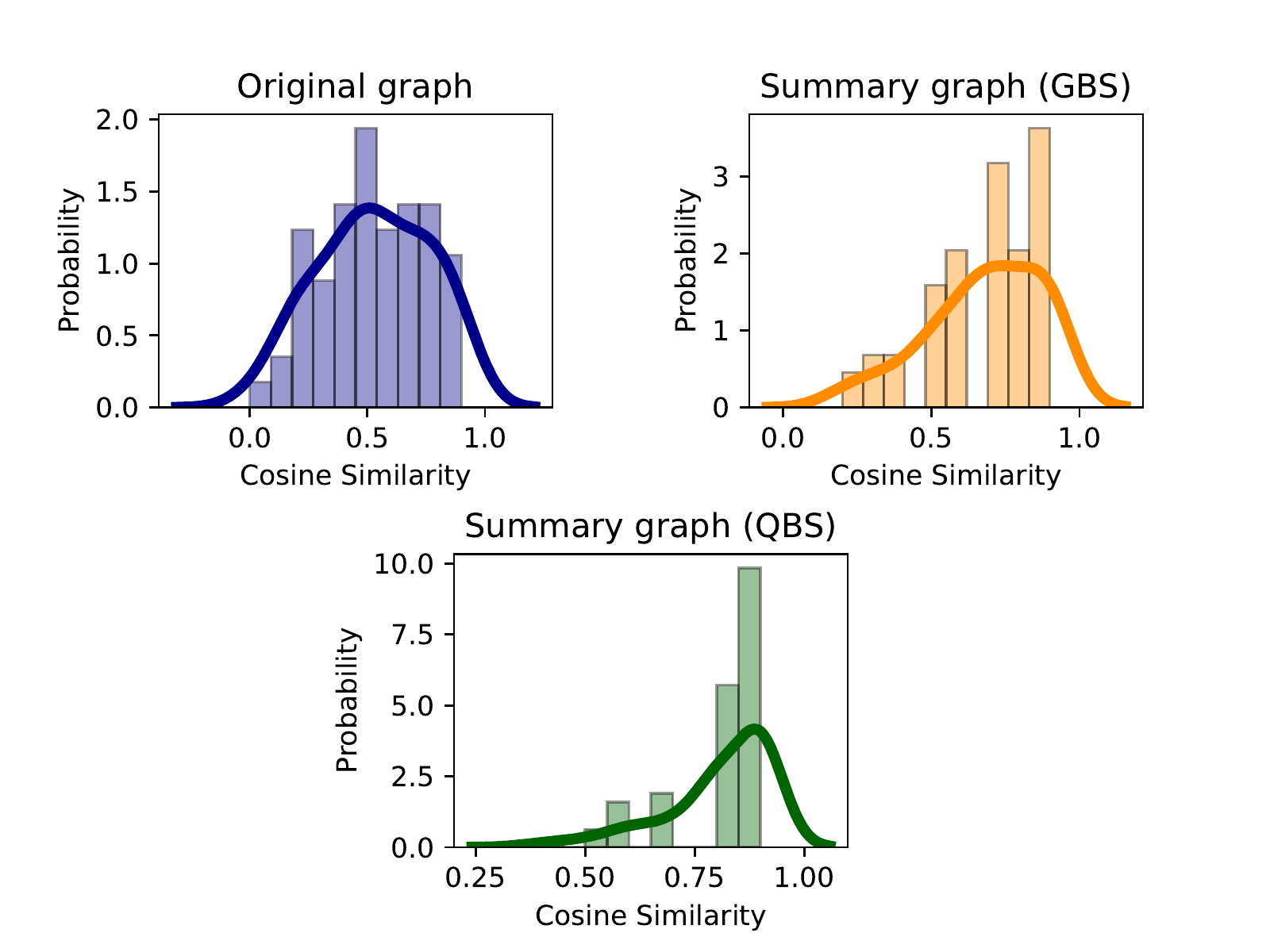}
     \label{fig:distplot2RDF2Vec}}
\hspace{0pt}\subfloat[WatDiv.10M]{
     \includegraphics[width=0.5\textwidth]{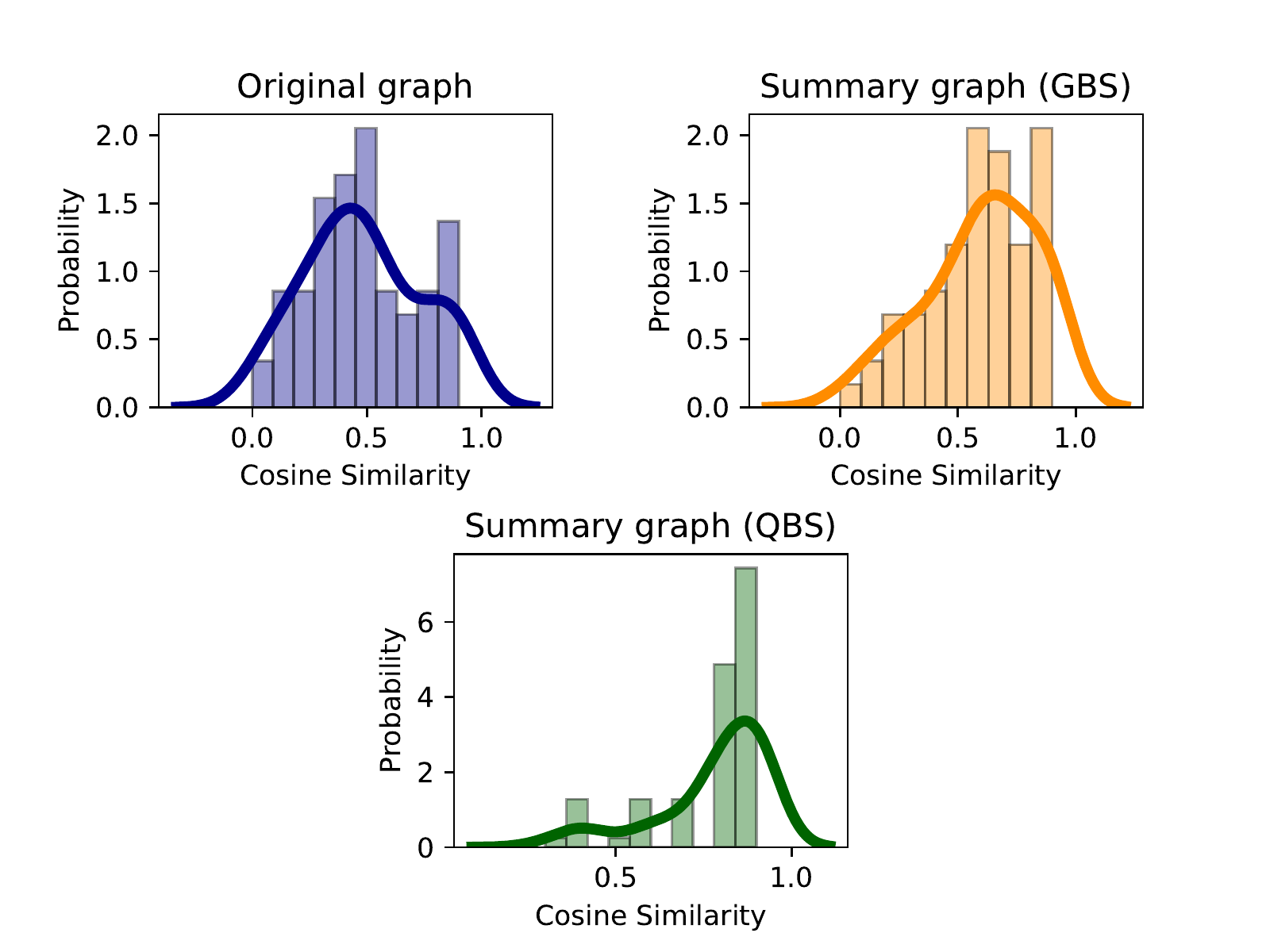}
     \label{fig:distplot3RDF2Vec}}
     \vspace{0pt}\subfloat[WatDiv.100M]{
     \includegraphics[width=0.5\textwidth]{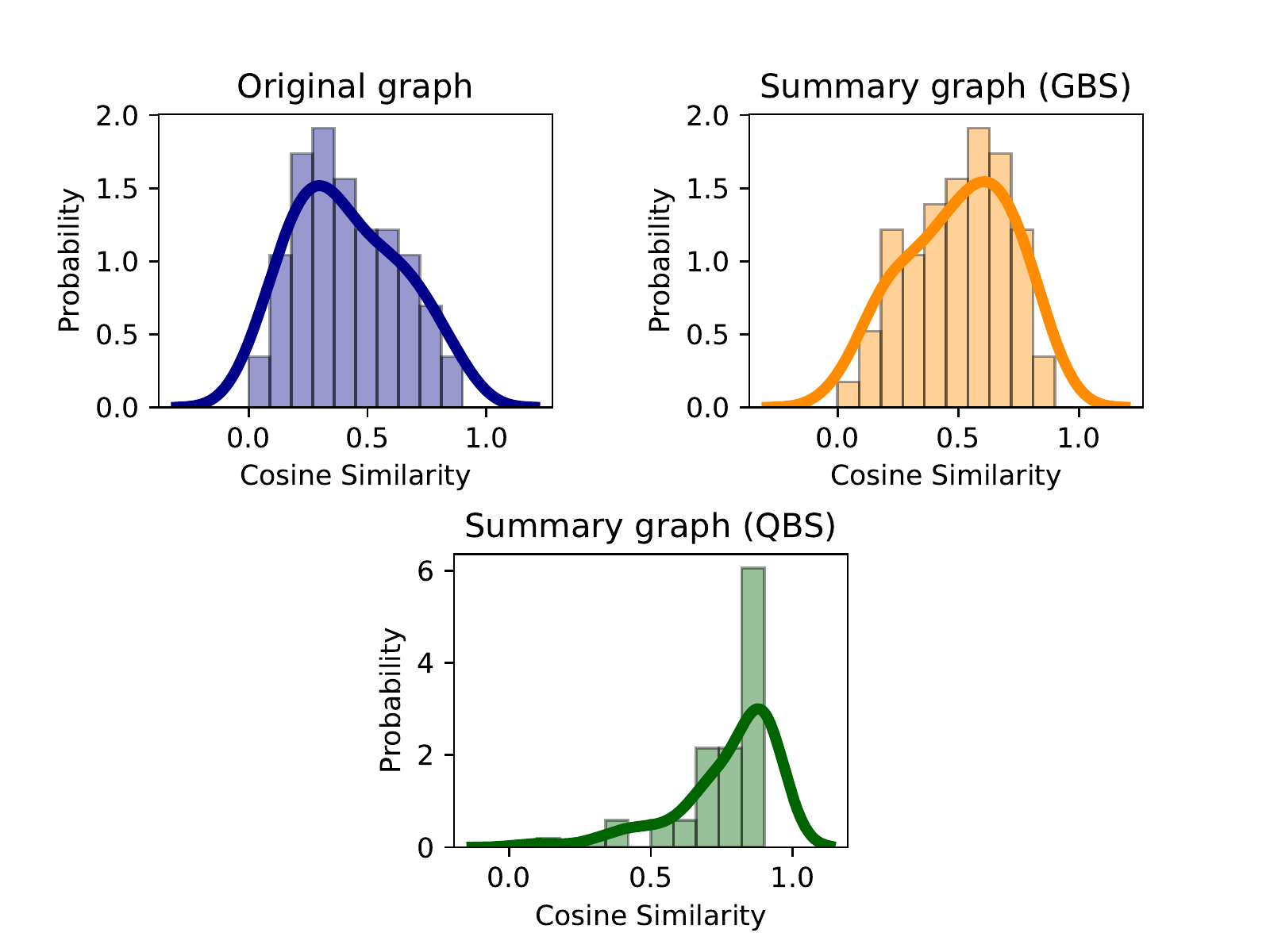}
     \label{fig:distplot4RDF2Vec}}
\caption{Compare the distribution of cosine similarity for predicate relatedness in the original and the summarized RDF graph based on GBS and QBS approaches using graph embedding model.}
 \label{fig:distributionRDF2Vec}
\end{figure*}

\subsection{Impact of predicate relatedness by analyzing similarity measures}

The distribution of similar predicates in each RDF graph using diverse embedding models is different. Referring to the research question \textbf{RQ1}, the predicate relatedness in each dataset has an impact on finding complete answers. For rewriting queries, the embedding techniques such as Word2Vec and RDF2Vec find similar predicates to the given predicate in the query to transform it. If the most similar and semantically relevant predicates are found, the transformation of a query would be more efficient in terms of finding the complete answers. Since finding the most similar predicates can affect verifying the cardinality property; it becomes an important issue on rewriting queries. During experiments, it is discovered that not only different datasets with various sizes have a different distribution of similarity values, but also the different distribution can be seen in diverse embedding models. It means not only the number of triples in RDF datasets can change the result of predicate relatedness, but also the techniques used for embedding have an impact. \autoref{fig:distribution} and \autoref{fig:distributionRDF2Vec} show the distribution of cosine similarity to find the most similar predicates between the original RDF graph and summarized RDF graphs provided by Grouping Based Summarization(GBS) approach and Query Based Summarization (QBS) approach using word embedding model and graph embedding model, respectively.

The results suggest that the way to find similar predicates depends on the size of the RDF graph and the technique for embedding. As seen, the probability of cosine similarity close to 1 in QBS approach is higher than the probability in the original graph and summary graph by GBS approach. The reason is, summary graphs in QBS have much fewer and more similar triples compared with the original ones. In QBS approach, only a part of the RDF graph related to the query has been considered, not the whole graph. Hence, by considering the small part of RDF graph includes triples with the same objects, similar and relevant predicates are close to each other and the distance between them is less. Therefore, by decreasing the number of triples in the dataset and keeping relevant triples to the query, the probability of finding the most similar predicates to a given one in its neighborhood is higher. In GBS approach, the probability of having a cosine similarity close to 0.5 using word embedding and close to 0.7 using graph embedding is higher. The reason is, after summarizing RDF graphs based on an explanation in \autoref{sub:twophase}, many relevant triples are lost. So, the embedding model finds the properties which are not completely relevant and similar. If less relevant similar predicates are found for rewriting the query, then query over the summarized RDF graph will return incomplete answers. Hence, in GBS approach, applying a query over the summarized graph does not return the complete answers.
Moreover, these observations are more obvious in \autoref{fig:distributionRDF2Vec} where RDF2Vec model used as graph embedding model to find similar properties. RDF2Vec employs different walking strategies such as Random Walk, NGram Walk, HALK Walk, Walklet Walk, and Anonymous Walk (\cite{strategy}). In this work, Random Walk is applied to extract the walks over the knowledge graphs. After finding the proper threshold, only predicates with cosine similarity greater than the threshold are considered as synonyms to rewrite the queries. 
The predicate relatedness by analyzing similarity measures in RDF graphs plays an important role in transforming the query to retrieve complete answers. In QBS approach, by finding similar predicates using different embedding models, the transformation of the query has been done efficiently and leads to retrieving complete results.

%Time for generating summary graph only for QBS based on 10 queries (depends on the query is different)
\begin{table}[t!]
\centering
\resizebox{1\textwidth}{!}{
\begin{tabular}{cccccc}
\toprule
% \rowcolor{color2}
\textbf & \scriptsize{\#Triples of} & \multicolumn{2}{c}\scriptsize{GBS Approach} & \multicolumn{2}{c}\scriptsize{QBS Approach} \\
\textbf {Dataset} & \scriptsize{original RDF graph} & \scriptsize{SR} & \scriptsize{ST} & \scriptsize{SR} & \scriptsize{ST} \\
\toprule
\scriptsize{DBpedia} & \scriptsize{2,047} & \scriptsize{32.2} \% & \scriptsize{36.8 s} & \scriptsize{99 \%} & \scriptsize{2.8 s} \\\midrule
\scriptsize{ESBM} & \scriptsize{6,584} &  \scriptsize{33 \%} &  \scriptsize{34.2 s} &  \scriptsize{98.8 \%} &  \scriptsize{2.2 s} \\\midrule
\scriptsize{WatDiv.10M} & \scriptsize{10,916,457} & \scriptsize{41.8 \%} & \scriptsize{100 s} & \scriptsize{96.6 \%} & \scriptsize{14.2 s} \\\midrule
\scriptsize{WatDiv.100M} & \scriptsize{108,997,714} & \scriptsize{57 \%} & \scriptsize{787.6 s} & \scriptsize{97.4 \%} & \scriptsize{53.6 s}  \\
\bottomrule
\end{tabular}
}
\caption{RDF Graphs Summarization Ratio (SR) and Summarization Time (ST)}
\label{table:ky02}    
\end{table}

\begin{table}[t!]
\centering
\begin{tabularx}{\columnwidth}{*{5}{X}}	
\toprule
\multicolumn{1}{l}{}& \multicolumn{3}{c}{\scriptsize{\# of answers}} 
% & \multicolumn{1}{c}{\scriptsize{\# of triples}}
\\
\cline{2-4} 
\rule{0pt}{8pt}
\multirow{2}{*}{$\longrightarrow$} & \scriptsize{\textbf{Sparklify}} & \scriptsize{\textbf{Sparklify+GBS}} & \scriptsize{\textbf{Sparklify+QBS}} 
% & \scriptsize{\textbf{summary graph by QBS }} 
\\

\midrule
\multirow{5}{*}{\rotatebox{90}{\scriptsize{\textbf{DBpedia}}}}
$Q1$ & \scriptsize{$5$} & \scriptsize{$5$} & \scriptsize{$5$} \\
\hspace{0.2cm} $Q2$ & \scriptsize{$5$} & \scriptsize{$5$} & \scriptsize{$5$}\\
\hspace{0.2cm} $Q3$ & \scriptsize{$9$} & \scriptsize{$9$} & \scriptsize{$9$}\\
\hspace{0.2cm} $Q4$ & \scriptsize{$6$} & \scriptsize{$6$} & \scriptsize{$6$} \\
\hspace{0.2cm} $Q5$ & \scriptsize{$8$} & \scriptsize{$8$} & \scriptsize{$8$} \\
\midrule
\multirow{5}{*}{\rotatebox{90}{\scriptsize{\textbf{ESBM}}}}
$Q6$ & \scriptsize{$3$} & \scriptsize{$3$} & \scriptsize{$3$} \\
\hspace{0.2cm} $Q7$ & \scriptsize{$2$} & \scriptsize{$2$} & \scriptsize{$2$}\\
\hspace{0.2cm} $Q8$ & \scriptsize{$5$} & \scriptsize{$5$} & \scriptsize{$5$} \\
\hspace{0.2cm} $Q9$ & \scriptsize{$8$} & \scriptsize{$8$} & \scriptsize{$8$} \\
\hspace{0.2cm} $Q10$ & \scriptsize{$6$} & \scriptsize{$6$} & \scriptsize{$6$} \\
\midrule
\multirow{5}{*}{\rotatebox{90}{\scriptsize{\textbf{WatDiv.10M}}}}
$Q11$ & \scriptsize{$86$} & \scriptsize{$54$} & \scriptsize{$86$}\\
\hspace{0.2cm} $Q12$ & \scriptsize{$56$} & \scriptsize{$30$} & \scriptsize{$56$}\\
\hspace{0.2cm} $Q13$ & \scriptsize{$99$} & \scriptsize{$64$} & \scriptsize{$99$} \\
\hspace{0.2cm} $Q14$ & \scriptsize{$3,834$} & \scriptsize{$2,043$} & \scriptsize{$3,834$} \\
\hspace{0.2cm} $Q15$ & \scriptsize{$59$} & \scriptsize{$26$} & \scriptsize{$59$} \\
\midrule
\multirow{5}{*}{\rotatebox{90}{\scriptsize{\textbf{WatDiv.100M}}}}
$Q16$ & \scriptsize{$59$} & \scriptsize{$31$} & \scriptsize{$59$} \\
\hspace{0.2cm} $Q17$ & \scriptsize{$554$} & \scriptsize{$206$} & \scriptsize{$554$} \\
\hspace{0.2cm} $Q18$ & \scriptsize{$983$} & \scriptsize{$489$} & \scriptsize{$983$}\\
\hspace{0.2cm} $Q19$ & \scriptsize{$38,895$} & \scriptsize{$21,089$} & \scriptsize{$38,895$}\\
\hspace{0.2cm} $Q20$ & \scriptsize{$146$} & \scriptsize{$85$} & \scriptsize{$146$} \\
\bottomrule
\end{tabularx}
{\caption{Compare the number of answers by querying multi triple patterns query over the original graph and transformed query as simple query over the summarized graph provided by Grouping Based Summarization (GBS) approach and Query Based Summarization (QBS) approach using the query engine Sparklify.}\label{table:ky04}}
\end{table}

\begin{table}[t!]
\centering
\begin{tabularx}{\columnwidth}{*{7}{X}}
\toprule
\multicolumn{1}{l}{}& \multicolumn{5}{c}{\scriptsize{Execution Time (seconds)}} 
% & \multicolumn{1}{c}{\scriptsize{\# of triples}}
\\
\cline{2-6} 
\rule{0pt}{8pt}
\multirow{2}{*}{$\longrightarrow$} & \tiny{\textbf{S(Total)}} & \tiny{\textbf{S+GBS(SUM)}} & 
\tiny{\textbf{S+GBS(QA)}} & 
\tiny{\textbf{S+QBS(SUM)}} & \tiny{\textbf{S+QBS(QA)}} 
% & \scriptsize{\textbf{summary graph by QBS }} 
\\
\midrule
\multirow{5}{*}{\rotatebox{90}{\scriptsize{\textbf{DBpedia}}}}
$Q1/Q'1$ & \scriptsize{$19$} & \scriptsize{$37$} & \scriptsize{$6$} & \scriptsize{$2$} & \scriptsize{$3$} \\
\hspace{0.2cm} $Q2/Q'2$ & \scriptsize{$18$} & \scriptsize{$39$} & \scriptsize{$6$} & \scriptsize{$3$} & \scriptsize{$5$}\\
\hspace{0.2cm} $Q3/Q'3$ & \scriptsize{$29$} & \scriptsize{$45$} & \scriptsize{$9$} & \scriptsize{$3$} & \scriptsize{$6$}\\
\hspace{0.2cm} $Q4/Q'4$ & \scriptsize{$18$} & \scriptsize{$30$} & \scriptsize{$5$} & \scriptsize{$3$} & \scriptsize{$5$}\\
\hspace{0.2cm} $Q5/Q'5$ & \scriptsize{$17$} & \scriptsize{$33$} & \scriptsize{$6$} & \scriptsize{$3$} & \scriptsize{$4$}\\
\midrule
\multirow{5}{*}{\rotatebox{90}{\scriptsize{\textbf{ESBM}}}}
$Q6/Q'6$ & \scriptsize{$32$} & \scriptsize{$45$} & \scriptsize{$9$} & \scriptsize{$3$} & \scriptsize{$5$}\\
\hspace{0.2cm} $Q7/Q'7$ & \scriptsize{$19$} & \scriptsize{$36$} & \scriptsize{$7$} & \scriptsize{$2$} & \scriptsize{$4$}\\
\hspace{0.2cm} $Q8/Q'8$ & \scriptsize{$20$} & \scriptsize{$31$} & \scriptsize{$5$} & \scriptsize{$2$} & \scriptsize{$6$} \\
\hspace{0.2cm} $Q9/Q'9$ & \scriptsize{$20$} & \scriptsize{$30$} & \scriptsize{$6$} & \scriptsize{$2$} & \scriptsize{$5$}\\
\hspace{0.2cm} $Q10/Q'10$ & \scriptsize{$19$} & \scriptsize{$29$} & \scriptsize{$5$} & \scriptsize{$2$} & \scriptsize{$4$}\\
\midrule
\multirow{5}{*}{\rotatebox{90}{\scriptsize{\textbf{WatDiv.10M}}}}
$Q11/Q'11$ & \scriptsize{$85$} & \scriptsize{$102$} & \scriptsize{$56$} & \scriptsize{$14$} & \scriptsize{$32$}\\
\hspace{0.2cm} $Q12/Q'12$ & \scriptsize{$79$} & \scriptsize{$92$} & \scriptsize{$47$} & \scriptsize{$11$} & \scriptsize{$28$}\\
\hspace{0.2cm} $Q13/Q'13$ & \scriptsize{$80$} & \scriptsize{$89$} & \scriptsize{$78$} & \scriptsize{$15$} & \scriptsize{$26$}\\
\hspace{0.2cm} $Q14/Q'14$ & \scriptsize{$82$} & \scriptsize{$123$} & \scriptsize{$158$} & \scriptsize{$13$} & \scriptsize{$24$}\\
\hspace{0.2cm} $Q15/Q'15$ & \scriptsize{$98$} & \scriptsize{$94$} & \scriptsize{$38$} & \scriptsize{$18$} & \scriptsize{$28$}\\
\midrule
\multirow{5}{*}{\rotatebox{90}{\scriptsize{\textbf{WatDiv.100M}}}}
$Q16/Q'16$ & \scriptsize{$321$} & \scriptsize{$802$} & \scriptsize{$274$} & \scriptsize{$88$} & \scriptsize{$109$}\\
\hspace{0.2cm} $Q17/Q'17$ & \scriptsize{$147$} & \scriptsize{$541$} & \scriptsize{$349$} & \scriptsize{$42$} & \scriptsize{$56$}\\
\hspace{0.2cm} $Q18/Q'18$ & \scriptsize{$192$} & \scriptsize{$627$} & \scriptsize{$482$} & \scriptsize{$41$} & \scriptsize{$86$}\\
\hspace{0.2cm} $Q19/Q'19$ & \scriptsize{$401$} & \scriptsize{$1,190$} & \scriptsize{$850$} & \scriptsize{$63$} & \scriptsize{$168$}\\
\hspace{0.2cm} $Q20/Q'20$ & \scriptsize{$182$} & \scriptsize{$778$} & \scriptsize{$321$} & \scriptsize{$34$} & \scriptsize{$105$}\\
\bottomrule
\end{tabularx}
{\caption{Compare the execution time in the original and the summarized RDF graph provided by Grouping Based Summarization (GBS) and Query Based Summarization (QBS) approaches. Q is the query over the original RDF graph and Q’ is the query over the summarized RDF graph. S represents as baseline Sparklify executed over the original RDF graph. }\label{table:ky05}}
\end{table}

\subsection{Effectiveness of proposed summarized graph}
To evaluate the effectiveness of proposed approaches and answer research questions \textbf{RQ2} and \textbf{RQ3}, the results of experiments are reported based on reducing the size of the RDF graph and returning the complete answers. As seen in Table~\ref{table:ky02}, the Summarization Ratio (SR) for QBS approach is higher than GBS approach. It means the number of triples in the summarized graph provided by QBS approach is less than the number of triples in the summarized graph provided by GBS approach. In QBS approach by considering a query of the user and the part of the original RDF graph which is related to this query, the Summarization Ratio (SR) has significantly increased compared with GBS approach. The compactness property has been verified in both approaches, but it has a better result in QBS approach. For example, the QBS approach makes the dump of DBpedia dataset \textit{99\%} smaller than the original dataset, while by GBS approach it gets only \textit{32.2\%} smaller than the original one. Therefore, the compactness based on QBS approach is more than GBS approach. Also, Table~\ref{table:ky02} provides the Summarization Time (ST) which is the time required to generate the summary graph in both approaches based on the given datasets. Here, it should be mentioned that summarizing the RDF graph in QBS approach is not only related to the size of dataset but also is based on the query of the user. Therefore, the size and time provided here for generating the summarization graph is the average of the size and time among different queries. Moreover, Table~\ref{table:ky04} shows the result of applying the SPARQL queries with multi triple patterns over the original RDF graph and compares them with the result of applying the transformation of these queries over summarized graphs provided by both approaches. Based on the results provided here the cardinality in GBS approach has been verified only for small datasets and not for medium and large datasets, while based on Theorem~\ref{theo2} in QBS approach all information related to the query will be found in the summarized graph. Therefore, in QBS approach, the cardinality for all size of datasets has been verified. 
In order to return the answers for SPARQL queries, Sparklify, the scalable component which is the default query engine in the SANSA Stack, has been used during all the experiments. First, the query is applied over the original RDF graph. Later, the transformation of this query which is a simple query with fewer triple patterns is applied to the summary graph. As seen in Table~\ref{table:ky04}, the number of answers retrieved in GBS is less than the number of answers from the original RDF graph which means during summarizing some information is lost while in QBS the number of answers retrieved from the summarized graph is equal to the number of answers from the original RDF graph. Thus, in QBS, the cardinality property has been verified.

\begin{figure*}[t!]
\centering  
\hspace{0pt}\subfloat[DBpedia]{
     \includegraphics[width=0.5\textwidth]{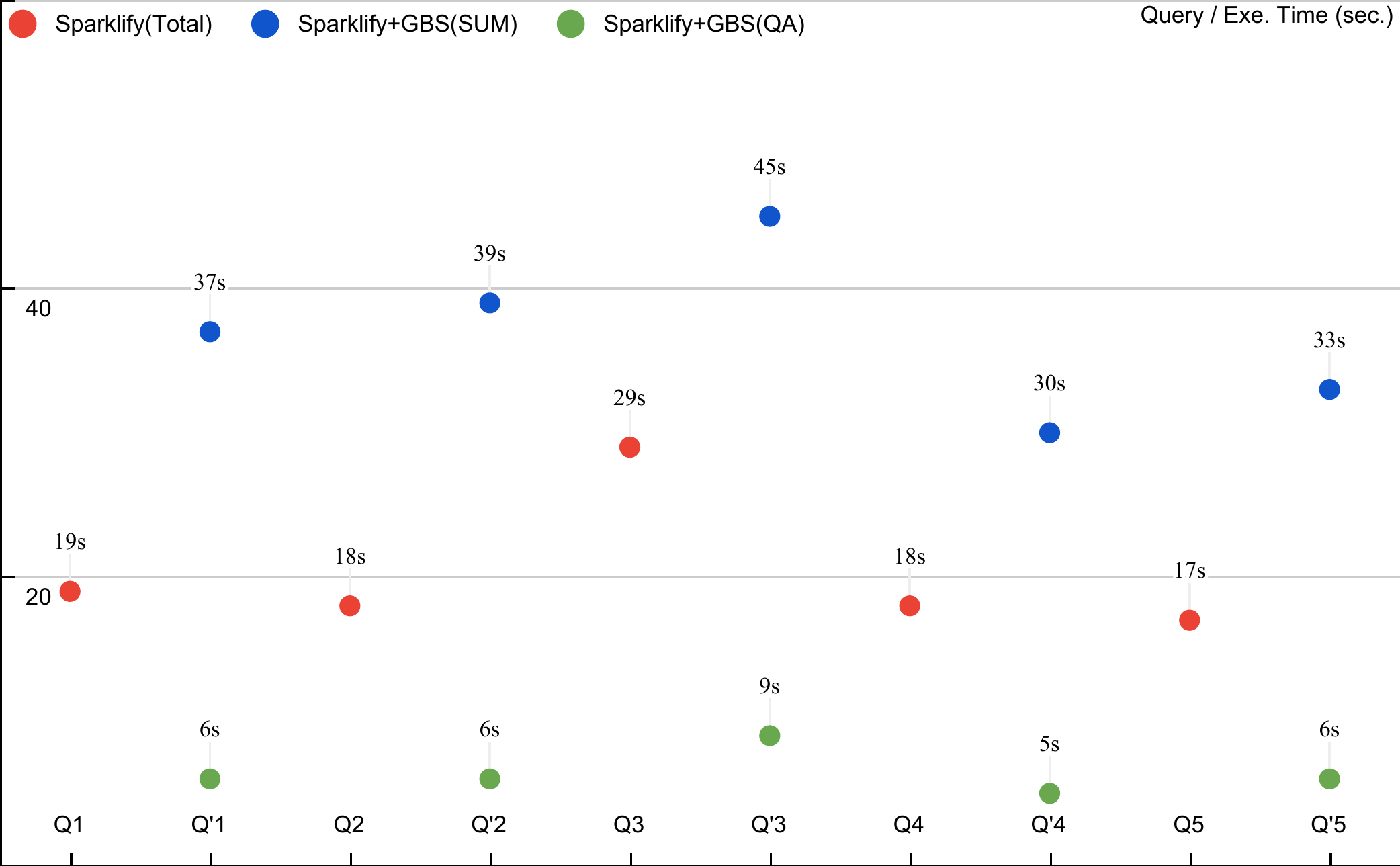}
     \label{fig:timeDBpedia}}
     \vspace{0pt}\subfloat[ESBM]{
     \includegraphics[width=0.5\textwidth]{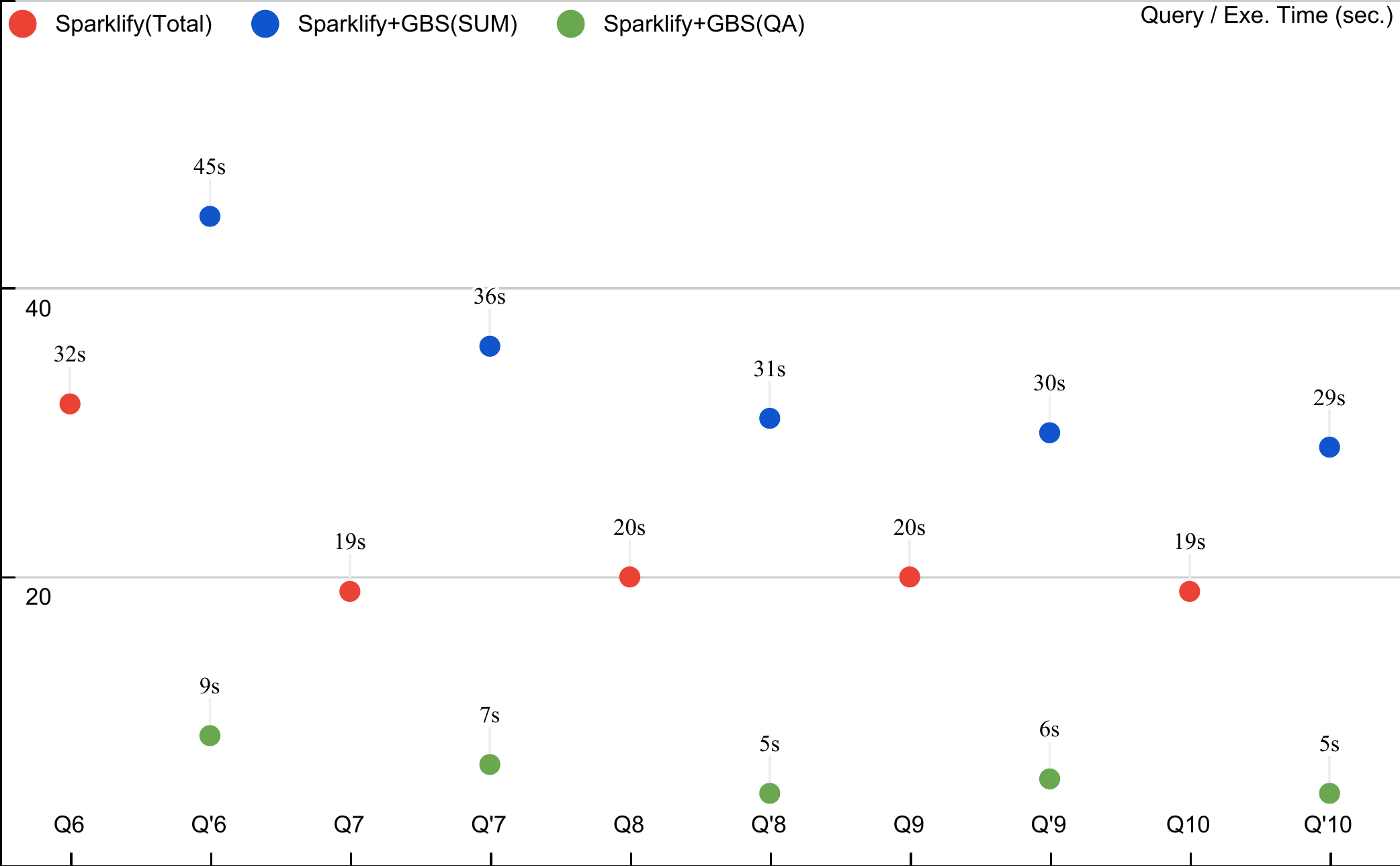}
     \label{fig:timeESBM}}
\hspace{0pt}\subfloat[WatDiv.10M]{
     \includegraphics[width=0.5\textwidth]{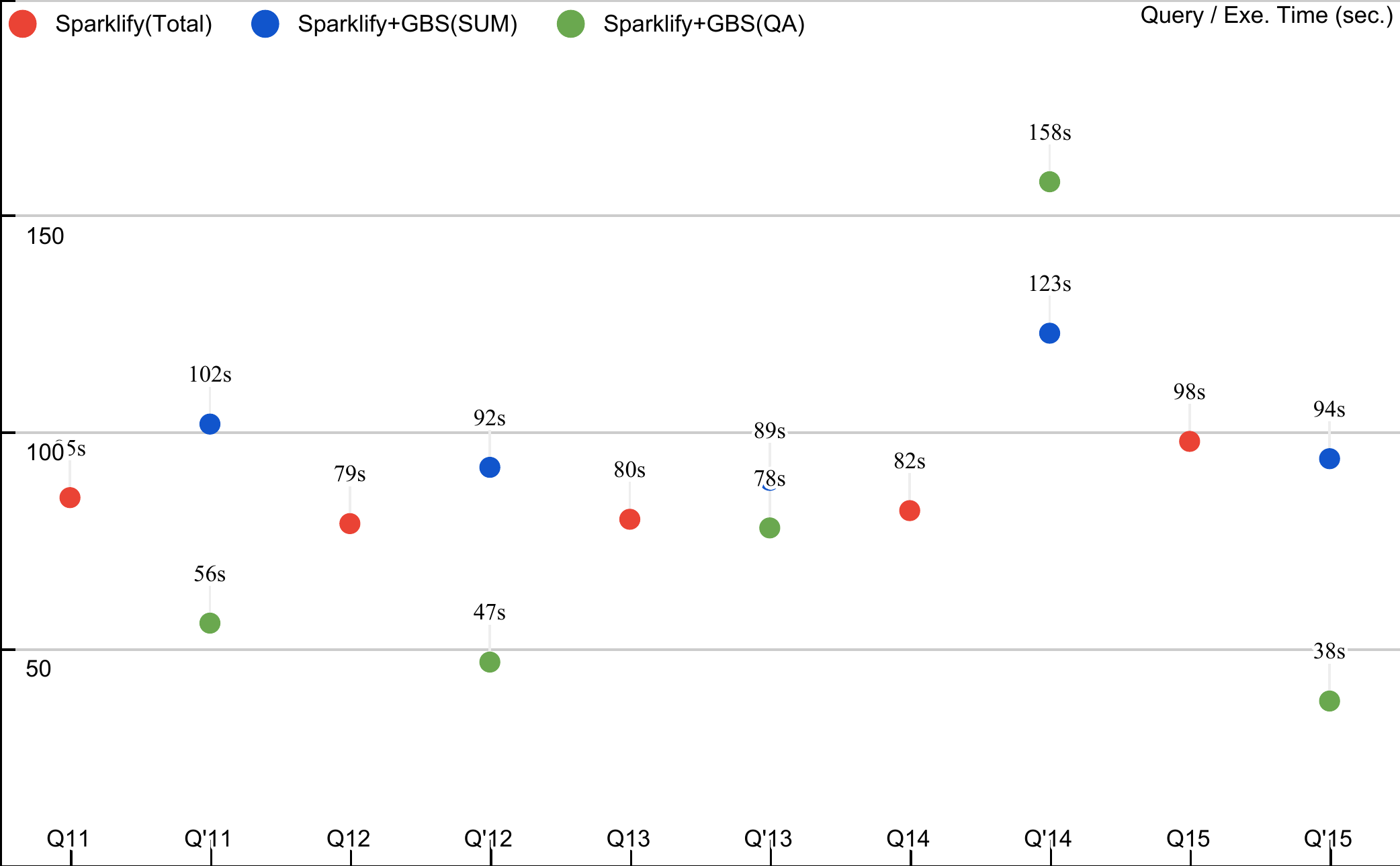}
     \label{fig:timeWat10}}
     \vspace{0pt}\subfloat[WatDiv.100M]{
     \includegraphics[width=0.5\textwidth]{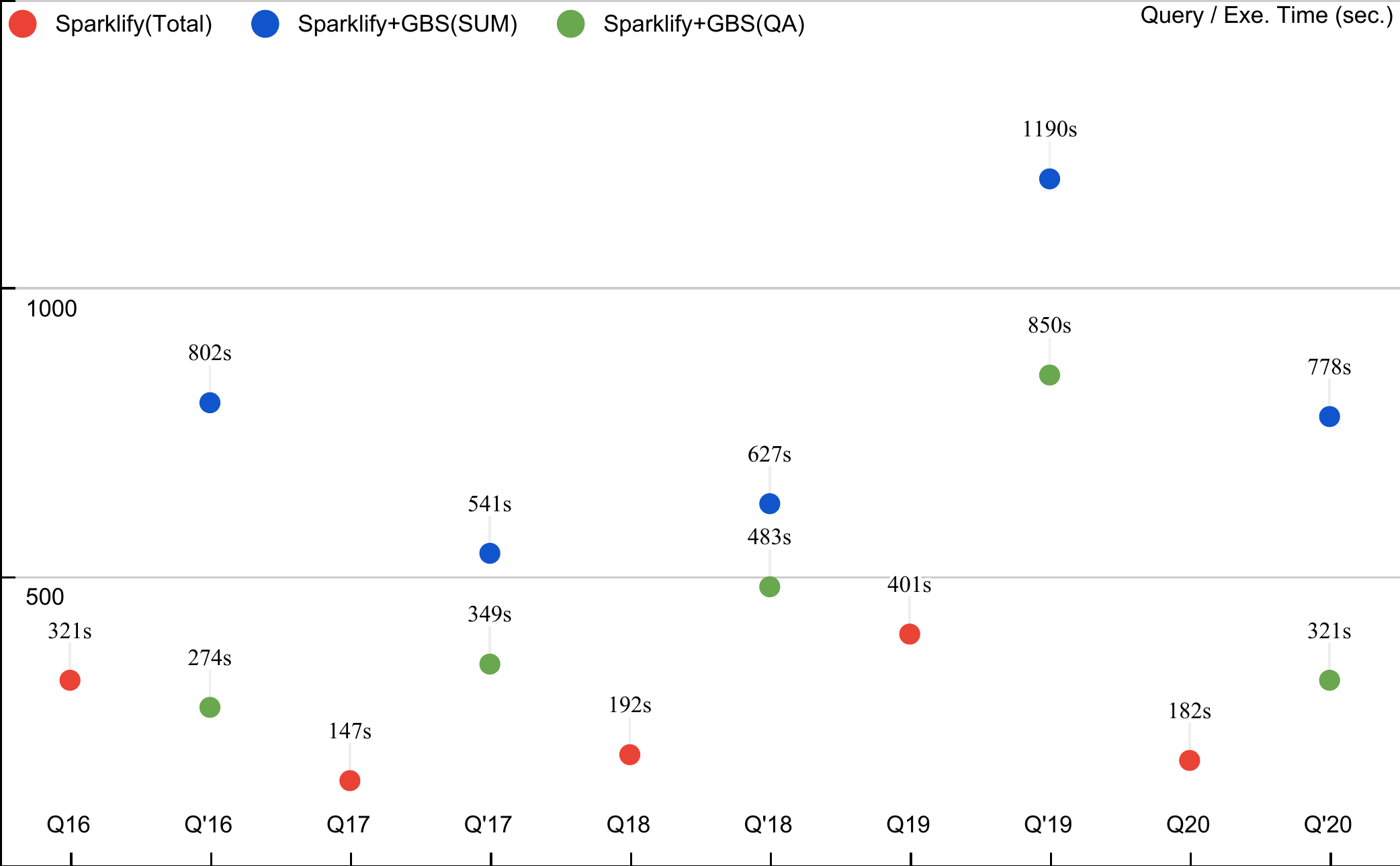}
     \label{fig:timeWat100}}
\caption{Comparing execution time in original RDF graph and summary RDF graph provided by Grouping Based Summary Graph (GBS) approach over four different datasets with different sizes. Q and Q' are executed over the original and summarized RDF graphs, respectively. Red dots show the execution time (in seconds) to answer a query over the original RDF graph. Blue and green dots show the time for summarizing RDF graph (SUM) and the time for query answering (QA) over summary graph, respectively. As seen, by increasing the size of datasets, execution time in the summarized graph is much higher than the original graph. }
 \label{fig:timerunning}
\end{figure*}

\begin{figure*}[t!]
\centering  
\hspace{0pt}\subfloat[DBpedia]{
     \includegraphics[width=0.5\textwidth]{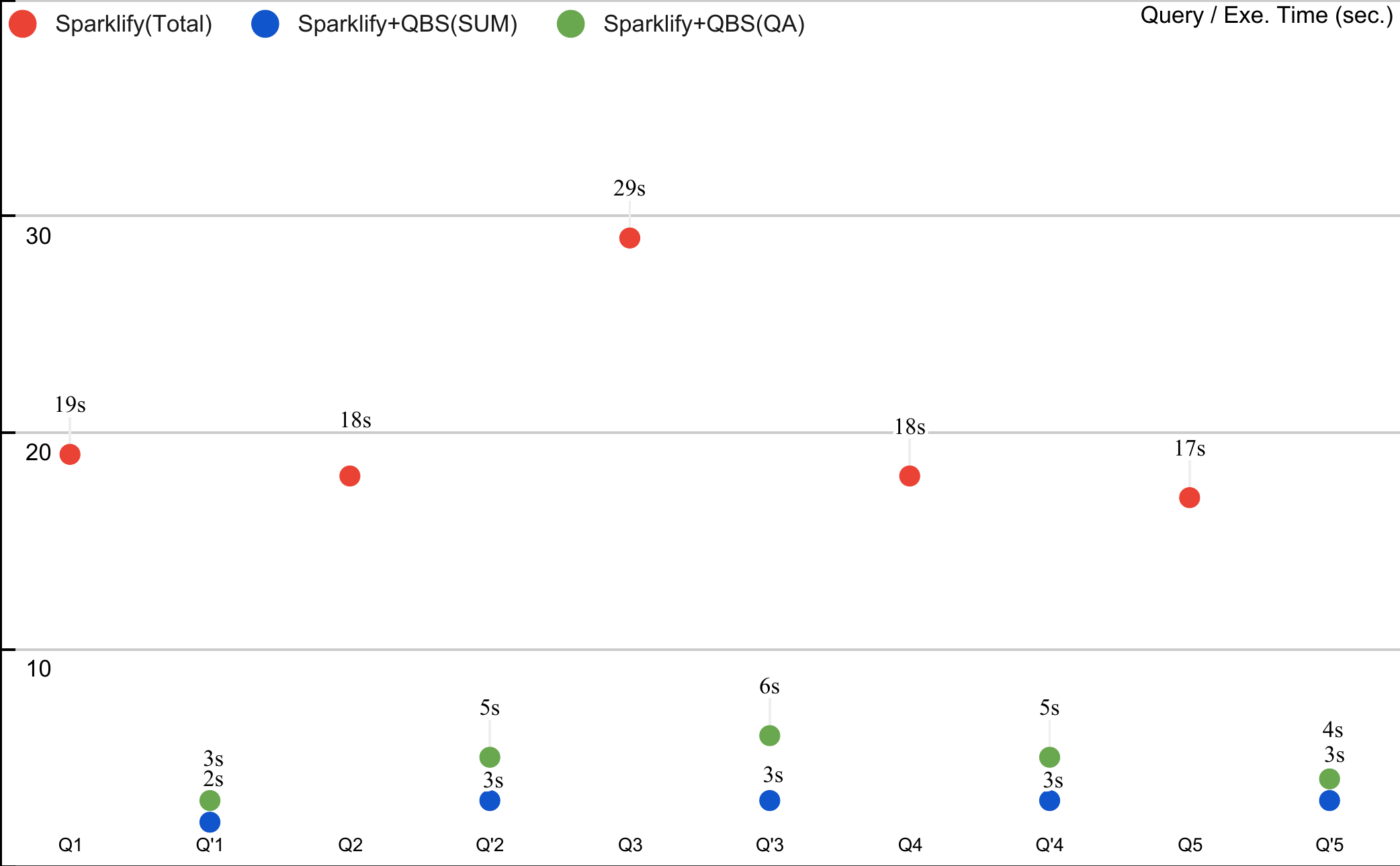}
     \label{fig:time2DBpedia}}
     \vspace{0pt}\subfloat[ESBM]{
     \includegraphics[width=0.5\textwidth]{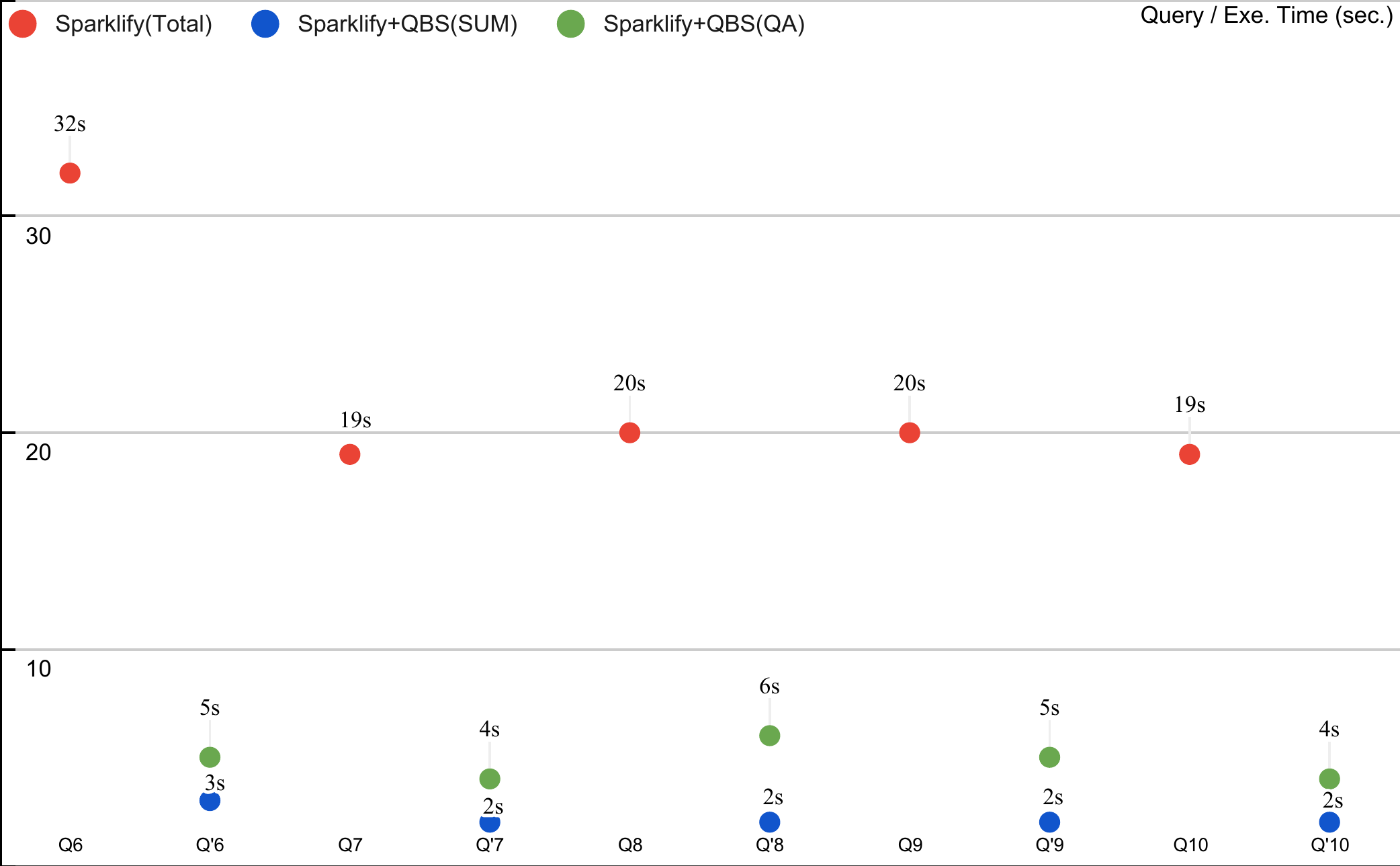}
     \label{fig:time2ESBM}}
\hspace{0pt}\subfloat[WatDiv10]{
     \includegraphics[width=0.5\textwidth]{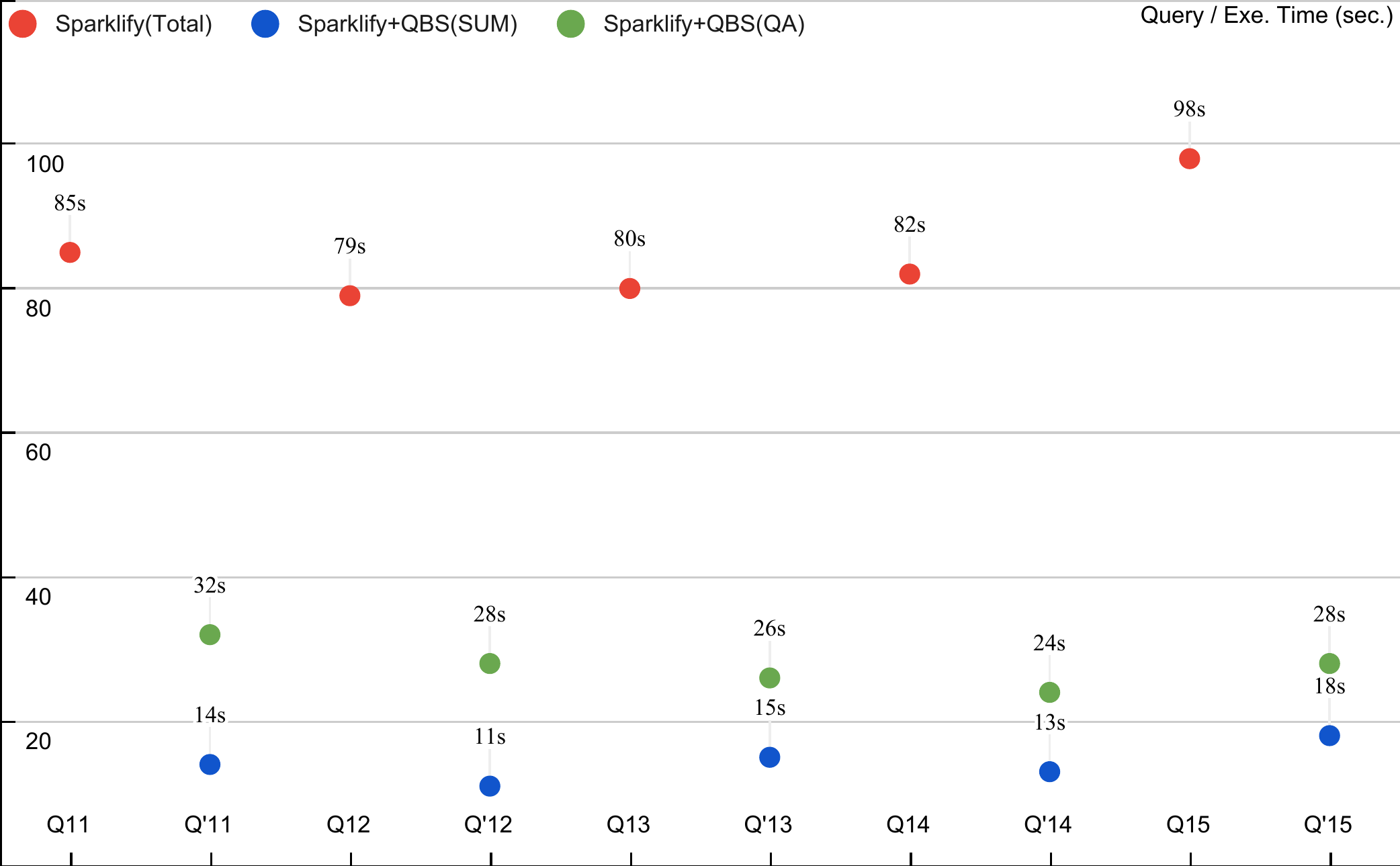}
     \label{fig:time2Wat10}}
     \vspace{0pt}\subfloat[WatDiv100]{
     \includegraphics[width=0.5\textwidth]{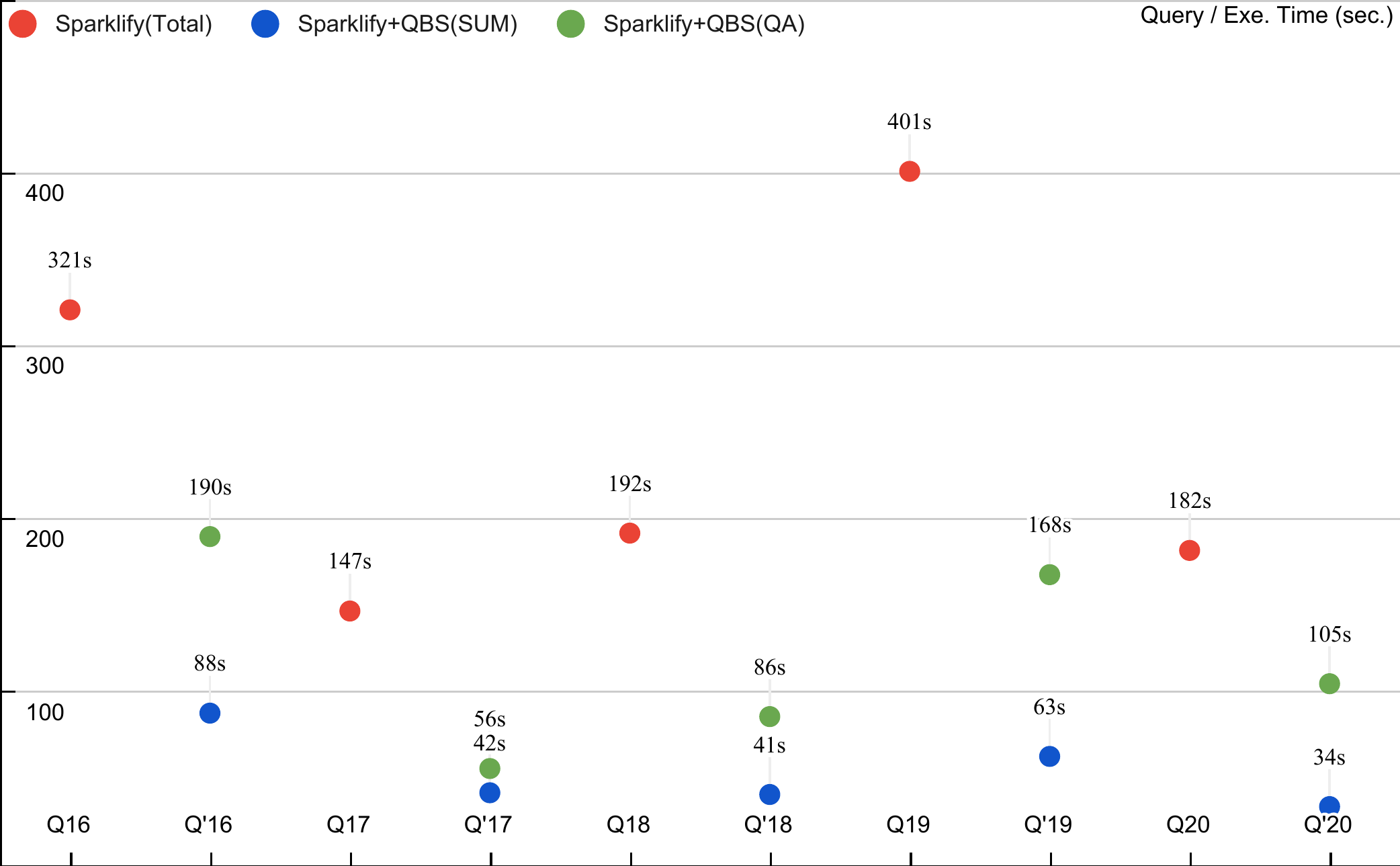}
     \label{fig:time2Wat100}}
\caption{Execution time in original and summary RDF graphs using Query Based Summary Graph (QBS) on four datasets. Q and Q' are queries over original and summarized RDF graphs, respectively. As seen, by increasing the size of datasets, execution time (in secs.) in the summarized graph is much less than in the original RDF graph. Thus, the execution time decreases when the dataset size on QBS increases.}
 \label{fig:timerunning2}
\end{figure*}

\subsection{Efficiency of transforming queries}
To answer the research question \textbf{RQ4}, we evaluate the execution time of query processing in the original RDF graph and the proposed summarized graphs provided by both approaches, GBS and QBS. The query processing has been done by the query engine, Sparklify. We report on the average query processing time after three times running the algorithms.

Table~\ref{table:ky05} shows the comparison of executing time over baseline and combing it with  proposed approaches. Also, to verify that the query engine Sparklify combined with QBS approach to query over the summarized RDF graph achieve a greater reduction in execution time than using the Sparklify over the original RDF graph, a Wilcoxon signed rank test is run with the result of p-value $< 6.103e-05$. As a result, the observed outcomes indicate a reduction in execution time of query processing over summarized RDF graphs generated by QBS. The execution time of Sparklify against the original RDF graph is the baseline. Moreover, both query processing time distributions differ significantly.

\autoref{fig:timerunning} illustrates the execution time of $20$ SPARQL queries over the original RDF graph compared with the summarized graph generated by Grouping Based Summary Graph (GBS) approach. The execution time is the summation of time for summarizing the RDF graph (SUM) and time for query answering (QA). As seen, execution time in summarized graph provided by GBS specially in larger datasets (Figure \autoref{fig:timeWat10} and Figure \autoref{fig:timeWat100}) is much higher than the execution time in original RDF graph. \autoref{fig:timerunning2} shows the same experiments but over summarized RDF graph generated based on Query Based Summary Graph (QBS) approach in four different size of datasets. As seen in \autoref{fig:timerunning2}, the execution time in summarized graph provided by QBS specially in larger datasets (Figure \autoref{fig:time2Wat10} and Figure \autoref{fig:time2Wat100}) is much less than the execution time in original RDF graph. The approach of QBS speeds up the execution time by up to $80 \%$ by summarizing the RDF graph. Therefore, the execution time of query processing in the summarized graph of QBS approach is much less than in the summarized graph of GBS approach. Also, by comparing the summarizing time (SUM) and querying answering (QA) in both approaches it is observed that the time for summarizing in QBS approach is less than the time for querying, while in GBS approach the summarizing time (SUM) is higher than query answering (QA) time. In general, all the observations show the optimization of the grouped based approach, GBS, to the query based approach, QBS, in terms of compactness, completeness, and execution time.

\section{Discussion}
\label{sec:discussion}
The presented experimental results confirm that the proposed RDF graph summarization technique is able to reduce the size of RDF graph by preserving necessary information. Moreover, the significant decrease in execution time is observed during query processing. 

\subsection{Contribution to literature}
The results of our summarization approaches are compared with the SPARQL query engine, Sparklify, as a baseline over the original RDF graph. The number of retrieved answers and the execution time over the summarized graph compared with the original RDF graph are shown in Table~\ref{table:ky04} and Table~\ref{table:ky05}, respectively. The results describe that the query processing over proposed summarized RDF graph is superior to querying over original RDF graph. Table~\ref{table:ky} provides an overview of existing summarization methods mentioned in \autoref{sec:related}, and compared with our approach. As seen, summarization methods in a large variety of concepts are different. However, QBS is able to not only reduce the size of the graph, but also speed up execution time. Thus, we define data management methods that can be used to enrich the portfolio of frameworks for managing and querying larger RDF graphs. Given the rapid increase of large RDF datasets, these methods will play a relevant role in scalability, providing thus the basis for the development of real-world applications.

\subsection{Practical implication}
The proposed approach is aimed to ensure the compactness, completeness, and improve execution time. It focuses on reducing the size of RDF graph by providing the lossless and low-cost query processing.
Although our approach performed better than the baseline, there are still some issues to discuss. For example, the assumption in our work is the original knowledge graphs consist of synonyms properties. In case the knowledge graph has few or no synonyms for properties, the summarization cannot be done properly. In the naive summarizing approach (GBS), the size of the RDF graph has vital impact on rate of summarization and time of execution. On the other hand, in the optimized QBS approach, the summarization ratio and execution time has direct relationship to the complexity of queries. Transformation of these complex queries to the simple ones based on similarity measures needs to be done efficiently. Defining the proper candidate sets helps queries be transformed in a way to retrieve the complete answers in less execution time. The techniques and methods to select and prune candidate sets still is an issue to discuss. However, our optimized summarization approach can be applied to any RDF knowledge graphs with synonymous properties. Each embedding technique has its own advantages and disadvantages on a variety of knowledge graphs. Except embedding techniques used in this paper, other techniques can be studied and compared. Moreover, other state-of-the-art summarization techniques with more evaluation criteria such as recall and precision to assess the accuracy of our approach can be added to the empirical evaluation. These limitations of our work needs to be addressed in future work.

\begin{table}[t!]
\centering
\resizebox{1\textwidth}{!}{
\begin{tabular}{cccccc}
\toprule
% \rowcolor{color2} 
\textbf{Research} & \textbf{Summary Type} & \textbf{Input} & \textbf{Technique} & \textbf{Output Graph} & \textbf{Purpose} \\
\toprule
ASSG by \cite{ASSG} & Structural RDF & Instance & Compression & Compressed
Graph & Query Answering\\\midrule
\cite{10.1145/2630602.2630610} & Structural RDF & Instance & Partitioning & RDF Graph &  Query Answering\\\midrule
\cite{Sydow} & Structural RDF & Instance & Selecting Sub-Graph & RDF Graph & Visualization \\\midrule
\cite{Zneika} & Pattern Mining & Instance & Approximate Graph Patterns & RDF Graph & Query Answering \\\midrule
\cite{Karim2020CompactingFS} & Pattern Mining & RDF Graph & Factorization & Factorized RDF Graph & Query Processing \\\midrule
CoSum by \cite{Ghasemi} & Statistical RDF & k-type & Grouping & Super & Entity Resolution\\\midrule
SPARQL Similarity Search & Hybrid RDF & Instance & Structural and & Multi-layer & Query Optimization \\
by \cite{Zheng:2016:SSS:2983200.2983201} & & and  Schema & Pattern Mining & Graph &  \\\midrule
GBS & Hybrid RDF & RDF Graph & Group based & Multi & Query Processing \\
Proposed Naive Approach & & & Summarization & Sub-Graphs \\\midrule
QBS & Hybrid RDF & RDF Graph & Query based & Selective  & Query Processing \\
Proposed Smart Approach & & & Summarization & RDF Triples\\
\bottomrule
\end{tabular}
}
\caption{Comparison of some graph summarization techniques.}
\label{table:ky}  
\end{table}

\section{Conclusions and future work}
\label{sec:conclusion}
In this paper, we tackled the challenge of RDF graph summarization to optimize query processing. The proposed techniques contribute to the portfolio of tools to efficiently manage knowledge graphs, which is of significant relevance given the role of knowledge graphs in knowledge representation.  
We presented the Grouping Based Summarization (GBS) and the Query Based Summarization (QBS) approaches. QBS optimizes GBS and ensures query completeness.
Technically, we implemented our solutions on top of the state-of-the-art SANSA Stack, allowing our methods to run on large-scale RDF graphs using Apache Spark as a process engine.
From the evaluation results, QBS  provides a good Summarizing Ratio (SR) from \textit{96\%} to \textit{99\%} in terms of the number of triples that are needed for evaluation.
From the query time point of view, there is a notable reduction in query execution, up to \textit{80\%} in the proposed Query Based Summarization (QBS) approach compared with the original RDF graph queried by Sparklify.
This less complexity in QBS has led to a significant improvement compared with the original RDF graph. 
With due attention to the investigations and results obtained from the experiments in this work, the importance of summarizing large-scale RDF graphs using the SANSA framework becomes more distinct. 

In the future, we will focus on extending the proposed summarization RDF graph with different types of data from structured to unstructured to retrieve complete results with a reduction in execution time. Also, we intend to evaluate the proposed method in other knowledge bases such as WikiData. 

\section*{Acknowledgments}
This work has been partially supported by the EU H2020 projects CLARIFY [grant number 875160] and PLATOON [grant number 872592]; and the EraMed project  P4-LUCAT [grant number 53000015].

\bibliographystyle{unsrtnat}
\bibliography{mybiblio} 

\end{document}